\documentclass[11pt]{article}

\usepackage{etex}
\usepackage[utf8x]{inputenc}
\usepackage[T1]{fontenc}
\usepackage{lmodern}
\usepackage{fullpage}
\usepackage{authblk}
\usepackage{caption,subcaption}
\usepackage{comment}
\usepackage{multirow}
\usepackage{cellspace}
\usepackage{slashbox}
\usepackage{amssymb}
\usepackage{amsmath}
\usepackage{amsfonts}
\usepackage{paralist}

\usepackage{amsthm}
\usepackage{latexsym}
\usepackage{wrapfig}
\usepackage{listings}
\usepackage{tabularx,multirow,mathtools}
\usepackage{relsize}
\usepackage{url}
 \newcommand\ForAuthors[1]
 {\par\smallskip                     
  \begin{center}
   \fbox
   {\parbox{0.9\linewidth}
    {\raggedright\sc--- #1}
   }
  \end{center}
  \par\smallskip                     %
 }
\usepackage[normalem]{ulem}
\newcommand{\commentaire}[1]{}

\newcommand{\assale}[1]{\textbf{{\color{blue}#1}}}
\newcommand{\ploc}[1]{\textbf{{\color{red}#1}}}

\usepackage{pgfplots}
\usetikzlibrary{patterns}

\pgfplotsset{plot coordinates/math parser=false}
\newlength\figureheight
\newlength\figurewidth 

\usepackage[linesnumbered]{algorithm2e}

\usepackage{xargs}                      
\usepackage[colorinlistoftodos,prependcaption,textsize=tiny]{todonotes}
\newcommandx{\ploccomment}[2][1=]{\todo[linecolor=red,backgroundcolor=red!25,bordercolor=red,#1]{#2}}
\newcommand{\iset}[1]{[#1]}

\newcommand{\nn}{\mathbb N}
\newcommand{\rr}{\mathbb R}
\newcommand{\rd}{\rr^d}
\newcommand{\br}{\overline{\mathbb R}}

\newcommand{\FS}{\mathcal{FS}\left(\pp,\br\right)}
\newcommand{\FSrun}{\mathcal{FS}\left(\pprun,\br\right)}
\newcommand{\F}{\mathcal{F}}
\newcommand{\fr}{\mathcal{F}\left(\pp,\br\right)}
\newcommand{\frrun}{\mathcal{F}\left(\pprun,\br\right)}
\newcommand{\frrunr}{\mathcal{F}\left(\pprun,\rr\right)}
\newcommand{\frp}{\mathcal{F}\left(\pp,\rr_{+}\right)}
\newcommand{\frr}{\mathcal{F}\left(\pp,\rr\right)}
\def\norm#1{\mbox{$\| #1 \|$}}

\newcommand{\st}{\operatorname{s.t.}}
\newcommand{\pp}{\mathbb P}
\newcommand{\pprun}{\pp_{\mathrm{run}}}

\newcommand{\rel}[1]{#1^{\mathcal R}}
\newcommand{\sha}[1]{#1^{\sharp}}
\newcommand{\becarre}[1]{#1^{\mathcal{R}}}

\newcommand{\mybrackets}[1]{\big(#1\big)}

\newcommand{\xlive}{X^0}
\newcommand{\op}{\Join}

\newcommand{\rea}{\mathfrak{R}}

\newcommand{\state}{\mathcal{X}}
\newcommand{\laws}{\mathcal{T}}
\newcommand{\xin}{X^{\mathrm{in}}}
\newcommand{\ind}{\mathcal{I}}

\newcommand{\Sol}{\mathrm{Sol}}

\newtheorem{theorem}{Theorem}[section]
\newtheorem{example}{Example}[section]
\newtheorem{definition}{Definition}[section]
\newtheorem{proposition}{Proposition}[section]
\newtheorem{corollary}{Corollary}[section]

\newtheorem{remark}{Remark}
\renewcommand{\phi}{\varphi}

\title{A Sums-of-Squares Extension of Policy Iterations}

\author[1]{Assal\'e Adj\'e\footnote{The work was done when the author was at IRIT, Universit\'e Paul Sabatier at Toulouse and was supported by the CIMI (Centre International de Math\'ematiques et d'Informatique) Excellence program ANR-11-LABX-0040-CIMI within the program
ANR-11-IDEX-0002-02 during a postdoctoral fellowship.}}
\author[2]{Pierre-Lo\"ic Garoche}
\author[3]{Victor Magron}
\affil[1]{
Laboratoire de Math\'ematiques et Physique (LAMPS),\authorcr
Universit\'e de Perpignan, Perpignan, France\authorcr
assale.adje@univ-perp.fr
}
\affil[2]{
Onera, the French Aerospace Lab, France.\authorcr
Universit\'e de Toulouse, F-31400 Toulouse, France.
\authorcr
pierre-loic.garoche@onera.fr
}
\affil[3]{
CNRS Verimag; 2 av de Vignate 38610 Gi\`eres France.\authorcr
victor.magron@imag.fr
}
\date{}

\begin{document}
\maketitle
\begin{abstract}
In order to address the imprecision often introduced by widening operators in static analysis,
policy iteration based on min-computations amounts to considering the
characterization of reachable value set of a program as an iterative computation of
policies, starting from a post-fixpoint. Computing each policy and the
associated invariant relies on a sequence of numerical optimizations. While the
early research efforts relied on linear programming (LP) to address linear properties of linear programs,
the current state of the art is still limited to the analysis of linear programs
with at most quadratic invariants, relying on semidefinite programming (SDP) solvers to compute 
policies, and LP solvers to refine invariants.

We propose here to extend the class of programs considered through the use of
Sums-of-Squares (SOS) based optimization. Our approach enables the precise analysis
of switched systems with polynomial updates and guards. The analysis presented
has been implemented in Matlab and applied on existing programs coming from the system control literature, improving both
the range of analyzable systems and the precision of previously handled ones.


\end{abstract}

\section{Introduction}
\label{introSOS}

A wide set of critical systems software including controller systems, rely on numerical computations. 
Those systems range from aircraft controllers, car engine
control, anti-collision systems for aircrafts or UAVs, to nuclear powerplant
monitors and medical devices such a pacemakers or insulin pumps.

In all cases, the software part implements the execution of an endless
loop that reads the sensor inputs, updates its internal states and controls
actuators. However the analysis of such software is hardly feasible with
classical static analysis tools based on linear abstractions. In fact, according
to early results in control theory from Lyapunov in the 19th century, a linear
system is defined as asymptotically stable iff it satisfies the Lyapunov criterion, i.e.~the
existence of a quadratic invariant. In this view, it is important to develop new analysis tools able to
support quadratic or richer polynomial invariants.


\if{
\begin{wrapfigure}{r}{.35\textwidth}
\vspace{-2em}
\begin{lstlisting}
init
while (cond0) {
  if (cond1) 
    x = poly1(x)
  elseif (cond2) 
    x = poly2(x)
  elseif (cond3) 
    x = poly3(x);
}
\end{lstlisting}
\caption{Programs considered.}
\vspace{-2em}
\label{fig:c_example}
\end{wrapfigure}
}\fi

\begin{figure}[!ht]
\begin{center}
\begin{subfigure}[b]{\textwidth}
\begin{center}
\begin{tabular}{|c|}
\hline
\begin{lstlisting}[mathescape=true]
x $\in$ $\xin$;
while ($r_1^0$(x)$\op$0 and ... and $r_{n_0}^0$(x)$\op$0){
  case ($r_1^1$(x)$\op$0 and ... and $r_{n_1}^1$(x)$\op$0): x = $T^1$(x);
  case ...
  case ($r_1^i$(x)$\op$0 and ... and $r_{n_i}^i$(x)$\op$0): x = $T^i$(x); 
}
\end{lstlisting}
\\
\hline
\end{tabular}
\end{center}
\caption{A one-loop program with a switch-case loop body.}
\label{fig:program}
\end{subfigure}

\vspace{0,5cm}

\begin{subfigure}[b]{\textwidth}
\begin{center}
\fbox{
\begin{minipage}{0.8\textwidth}
\begin{center}
Let $\xlive=\{y\mid r_1^0(y)\op 0 \text{ and } \ldots \text{ and } r_{n_0}^0(y)\op 0\}$

and for all cases $i$, $X^i=\{y\mid r_1^i(y)\op 0 \text{ and } \ldots \text{ and } r_{n_i}^i(y)\op 0\}$.

We can define the discrete-time switched system:
\[
x_0\in\xin,\ \forall\, k\in\nn,\ x_{k+1}=T(x_k), \text{ where } T(x)=T^i(x) \text{ if } x\in X^i\cap \xlive
\]
\end{center}
\end{minipage}
}
\end{center}
\caption{The discrete-time switched system correspondence of the program.}
\label{fig:switched}
\end{subfigure}
\end{center}
\caption{One-loop programs with switch-case loop body and its representation as a switched system.}
\label{fig:c_example}
\end{figure}

While most controllers are linear, their actual implementation is always more
complex. e.g. in order to cope with safety, additional constructs such as
antiwindups or saturations are introduced. Another classical approach is the use
of linear parameter varying controllers (LPV) were the gains of the linear
controller vary depending on conditions: this becomes piecewise polynomial controllers at
the implementation level.

We are interested here in bounding the set of reachable values of such
controllers using sound analyses, that is computing a sound over approximation
of reachable states.  We specifically focus on a class of programs larger  
than linear systems: constrained piecewise polynomial systems. 

  
  A loop is performed
  while a conjunction of polynomial inequalities\footnote{$\op$ is either the
    strict $<$ or the weak ($\leq$) comparison operator over reals.} is
  satisfied. Within this loop, different polynomial updates are performed
  depending on conjunctions of polynomial constraints. This class of programs is
  represented in Fig.~\ref{fig:program}. The program can be analyzed through its switched 
  system representation (see Fig.~\ref{fig:switched}). In the obtained system, the conditions to switch 
  are only governed by the state variable: at each time $k$, we consider the dynamics $T^i$ such that $x_k\in X^i$.
  So, the switch conditions do not depend on the time (we do not consider the time when we reach $X^i$) and are not defined 
  from random variables.
  
  From the point-of-view of the dynamical systems or control theory, to find an over-approximation of reachable values 
  set of a program of the class described in Fig.~\ref{fig:program} can be reduced to find a positive invariant of its 
  switched system representation (see Fig.~\ref{fig:switched}). 
  
  Moreover, the class of switched systems where the switch conditions only depend on the state variable is classical in control theory and includes the nonlinear systems with dead-zone, saturation, resets or hysteresis.



\subsection*{Related Works}

Reachability analysis is a long-standing problem in dynamical systems theory, especially when the system dynamics is nonlinear. In the particular case of polynomial systems, this problem has recently attracted several research efforts.

In~\cite{Wang13}, the authors use a method based on sublevel sets of polynomials to analyze the reachable set of a continuous time polynomial system with initial/general constraints being encoded by semialgebraic sets. Their method relies on a so-called iterative {\em advection algorithm} based on Sums-of-Squares (SOS) and semidefinite programming (SDP) to compute either inner of outer approximation of the backward reachable set, also named {\em domain of attraction} (DoA).  
The stability analysis of continuous-time hybrid systems with SOS certificates was investigated in~\cite{Robust09}.~\cite{Posa16} have recently applied analogous techniques to perform stability analysis and controller synthesis in the context of robotics.
Other studies rely on SOS reinforcement and moment relaxations  to obtain  hierarchies of approximations converging to sets of interest such as the DoA in continuous-time, either from outside in~\cite{henrion2014convex} or from inside in~\cite{KHJ12innerROA}. This approach relies on a linearization of the ordinary differential equation involved in the polynomial system, based on  Liouville's Equation satisfied by adequate occupation measures. This framework has been extended to hybrid systems in~\cite{shia14brs}, as well as to synthesis of feedback controllers in~\cite{majumdar2014convexArxiv}. Liouville's Equation can also be used to approximate other sets of interest, such as the maximal controlled invariant for either discrete- or continuous-time systems (see~\cite{KHJ14mci}).

By contrast with these prior works, our approach focuses on computing over approximation of the forward reachable set of a discrete-time polynomial system with finitely many guards, thanks to an algorithm relying on SOS and template policy iterations.

Template abstractions were introduced
in~\cite{DBLP:conf/sas/SankaranarayananSM04} as a way to define an abstraction
based on an a-priori known vector of templates, i.e. numerical expressions over
the program variables. An abstract element is then defined as a vector of reals
defining bounds $b_i$ over the templates $p_i$: $p_i(x_1, \ldots, x_d) \leq
b_i$.

Initially templates were used in the classical abstract interpretation setting to 
compute Kleene fixpoints with linear functions $p_i$. Typically,
the values of the bound $b_i$ increases during the fixpoint computation until
convergence to a post-fixpoint. 

Later in~\cite{DBLP:conf/esop/AdjeGG10} the authors proposed to consider generalized templates including
 quadratic forms and compute directly the
fixpoint of these template-based abstractions using numerical optimization. When
considering the sub-class of linear programs, fixpoints can
be computed by using a finite sequence of linear (LP) and semidefinite (SDP) optimization problems. 


Two dual approaches, respectively called Max-policies and Min-policies, can be applied. Max-policies~\cite{DBLP:conf/esop/GawlitzaS07} iterate from initial states
and compute \emph{policies} as relaxations through rewriting of an optimization
problem (forgetting about rank conditions). Min-policies~\cite{DBLP:conf/esop/GaubertGTZ07,DBLP:conf/esop/AdjeGG10} rely on duality
principle and determine a policy through the computation of a Lagrange multiplier. 


\subsection*{Contributions}
The present paper is a followup of~\cite{SAS1}, in which Sums-of-Squares (SOS) programming is used to analyze properties (such as boundedness or safety) of piecewise discrete polynomial systems.
The main contribution is the extension of the Min-policy iteration algorithm to improve the precision of the analysis of boundedness for such systems. 
Here, we develop a policy iteration on template domains based on polynomials.
The current approach improves the previous frameworks developed in~\cite{DBLP:conf/esop/GaubertGTZ07,DBLP:journals/corr/abs-1111-5223} to handle affine (or very specific) piecewise affine systems with quadratic templates and semidefinite programming. 
This improvement comes from the use of SOS programming to develop an SOS extension 
of a relaxed functional, sharing the same properties as the one defined in~\cite{DBLP:journals/corr/abs-1111-5223}.
In particular, we show, as in~\cite{aadje_nsv}, that this policy iteration has the desired convergence guarantees.


\subsection*{Organization of the paper}

The paper is structured as follows: we characterize the class of programs
considered -- constrained piecewise discrete-time polynomial systems -- together with  
 their collecting semantics as a least fixpoint. Then, in
Section~\ref{sec:basicinductive-domainsSOS}, we recall definitions of generalized templates, their
expression as an abstract domain and the definition of the abstract transfer
function. Section~\ref{sec:relaxedSOS} proposes an abstraction of the transfer
function using an SOS reinforcement, while Section~\ref{sec:policySOS} relies on this
abstraction to perform policy iteration. Finally Section~\ref{sec:exampleSOS}
presents experiments.



\section{Constrained Piecewise Discrete-Time Polynomial Systems: Definition and
  Collecting Semantics}
\label{sec:contextSOS}
In this paper, we are interested in proving automatically that the set of values taken by the variables 
of the analyzed program is bounded. In the rest of the paper, we analyze 
the program by decomposing it using its dynamic system representation. The boundedness problem is thus reduced to prove
that the set of all the possible trajectories of a dynamical system is bounded. 
Since the analyzed program has the form depicted by Figure~\ref{fig:program},
we consider the special class of discrete-time dynamical systems introduced in Figure~\ref{fig:switched} that is:
\begin{enumerate}[(i)]
\item their dynamic law $T$ is a piecewise polynomial function, and
\item the state variable $x$ is constrained to live in some given basic
  semi-algebraic set\footnote{For instance membership of the sub-level set $\{x \in \rd \mid 1 - \|x\|_2^2 \leq 0\}$, thus this does not entail boundedness of variable values.}.
\end{enumerate} 

We recall that a set is a basic semi-algebraic set if
and only if it can be represented as a conjunction of strict or weak
polynomial inequalities (``basic'' means that no disjunction occurs).
Note that $T$ is a piecewise polynomial function with respect to a given
partition, meaning that if we restrict $T$ to be an element of the partition then $T$
is a polynomial function. 

We now give a formal definition of \emph{constrained piecewise discrete-time polynomial system} (PPS for short).

First to define a PPS, we need a partition. 
Let $\ind$ be a finite set of partition labels and
$\state=\{X^i \subseteq \rd \mid i\in\ind\}$ be a partitioning, that is a given family of
basic semi-algebraic sets satisfying the following: 
\begin{equation}
\label{partition}
\bigcup_{i\in\ind} X^i=\rd,\ \forall\, i,j\in\ind,\ i\neq j \Rightarrow X^i\cap X^j= \emptyset \enspace .
\end{equation}
Each set $X^i$ of the partition corresponds to a case in the loop body of the program given in Figure~\ref{fig:program} as the cases are assumed to be disjoint.
By definition of basic semi-algebraic sets, it follows that for all $i\in\ind$, there exists a
family of $n_i$ polynomials $\{r_j^i,j\in \iset{n_i}\}$ such that:
\begin{equation}
\label{semialgebraic}
X^i=\left\{x\in\rd \mid r_j^{i}(x) \op 0\ \forall\, j\in \iset{n_i}\right\} \,.
\end{equation}
where $\op$ is either $<$ or $\leq$ and $\iset{n_i}$ denotes the set $\{1,\ldots,n_i\}$.
Eq.~\eqref{semialgebraic} matches with the program given in Figure~\ref{fig:program}. 
Guards are defined as conjunctions of polynomial inequalities and thus are basic semi-algebraic sets.

The second tool needed to define a PPS is the piecewise polynomial dynamic relative to the partition. 
Let $T:\rd\mapsto \rd$ be a piecewise polynomial function w.r.t. to the
partitioning $\state$. By definition, there exists a family of polynomials
$\{T^i:\rd\mapsto \rd,i\in \ind\}$ such that for all $i\in\ind$:

\begin{equation}
\label{piecewise}
T(x)=T^i(x),\ \forall x\in X^i\enspace .
\end{equation}
Eq.~\eqref{piecewise} matches with the polynomial updates in Figure~\ref{fig:program}.

Finally, it remains to define the initialization and the set where the state variable lies. Let $\xin$ and $\xlive$ be two basic semi-algebraic sets of $\rd$, $\xin$
supposed to be compact, i.e. closed and bounded. The two sets can be represented
as in Eq.~\eqref{semialgebraic} using their respective family of
$n_{\mathrm{in}}$ and $n_0$ polynomials:
\[
\xin=\left\{x\in\rd \mid r_j^{\mathrm{in}}(x) \op 0\ \forall\, j\in\iset{n_{\mathrm{in}}}\right\}
\text{ and } 
\xlive=\left\{x\in\rd \mid r_j^{0}(x) \op 0\ \forall\, j\in \iset{n_{0}}\right\} \enspace ,
\]
where for all $j\in\iset{n_\mathrm{in}}$, $r_j^\mathrm{in}:\rd\mapsto \rr$ and for all $k\in\iset{n_0}$, 
$r_k^0:\rd\mapsto \rr$ is a polynomial.

The set $\xin$ and $\xlive$ respectively denote the set of initial states of the program and the set 
which defines the loop condition in Figure~\ref{fig:program}.

Let $\state$ be the family of sets $\{X^i, i\in \ind\}$ satisfying Eq.~\eqref{partition} and $\laws$ be the family of functions $\{T^i, i\in\ind\}$ satisfying Eq.~\eqref{piecewise}.
We define the PPS associated to the quadruple $(\xin,\xlive,\state,\laws)$
as the system satisfying the following dynamic:
\begin{equation}
\label{pws}
x_0\in \xin, \text{ and } \forall\, k\in\nn, \text{if } x_k\in X^0,\ x_{k+1}=T(x_k) \,.
\end{equation}
In the rest of the paper, a PPS dynamical system is identified by a quadruple $(\xin,\xlive,\state,\laws)$.
\begin{example}[Running example]
\label{running}
We consider the following running example corresponding to a slightly modified version
of~\cite[Example 3]{DBLP:conf/cdc/AhmadiJ13}. By comparison with~\cite[Example 3]{DBLP:conf/cdc/AhmadiJ13}, the semi-algebraic set $X^1$ (resp.~$X^2$) is introduced to represent conditions under which we use the polynomial update $T^1$ (resp.~$T^2$). 
The PPS is the quadruple $(\xin,X^0,\{X^1,X^2\},\{T^1,T^2\})$, where:\\
\[
\begin{array}{lcl}
  \begin{array}{l}
    \xin= [-1, 1] \times [-1, 1]\\
    X^0=\rr^2 
  \end{array}
  &  \text{ and }& 
  \left\{
    \begin{array}{l}
      X^1=\{x\in\rr^2\mid -x_1^2+1\leq 0\} \\
      X^2= \{x\in\rr^2\mid x_1^2-1< 0\}    
    \end{array}
  \right.
\end{array}
\]
  \\
\noindent
and the family of functions $\{T^1,T^2\}$, defined as follows:
\[
\begin{array}{c}
T^1(x_1,x_2)=\begin{pmatrix} 
0.687x_1+0.558x_2-0.0001x_1 x_2 \\ 
-0.292x_1+0.773x_2
\end{pmatrix}\\
\text{and}\\
T^2(x_1,x_2) =\begin{pmatrix} 
0.369x_1+0.532x_2-0.0001x_1^2\\
-1.27x_1+0.12x_2-0.0001x_1x_2
\end{pmatrix}
\end{array}
\]
Its (discretized) reachable value set is simulated and depicted at Figure~\ref{fig:running_ex_traces}.
\begin{figure}[ht!]
 \centering
 \includegraphics[width=.5\textwidth]{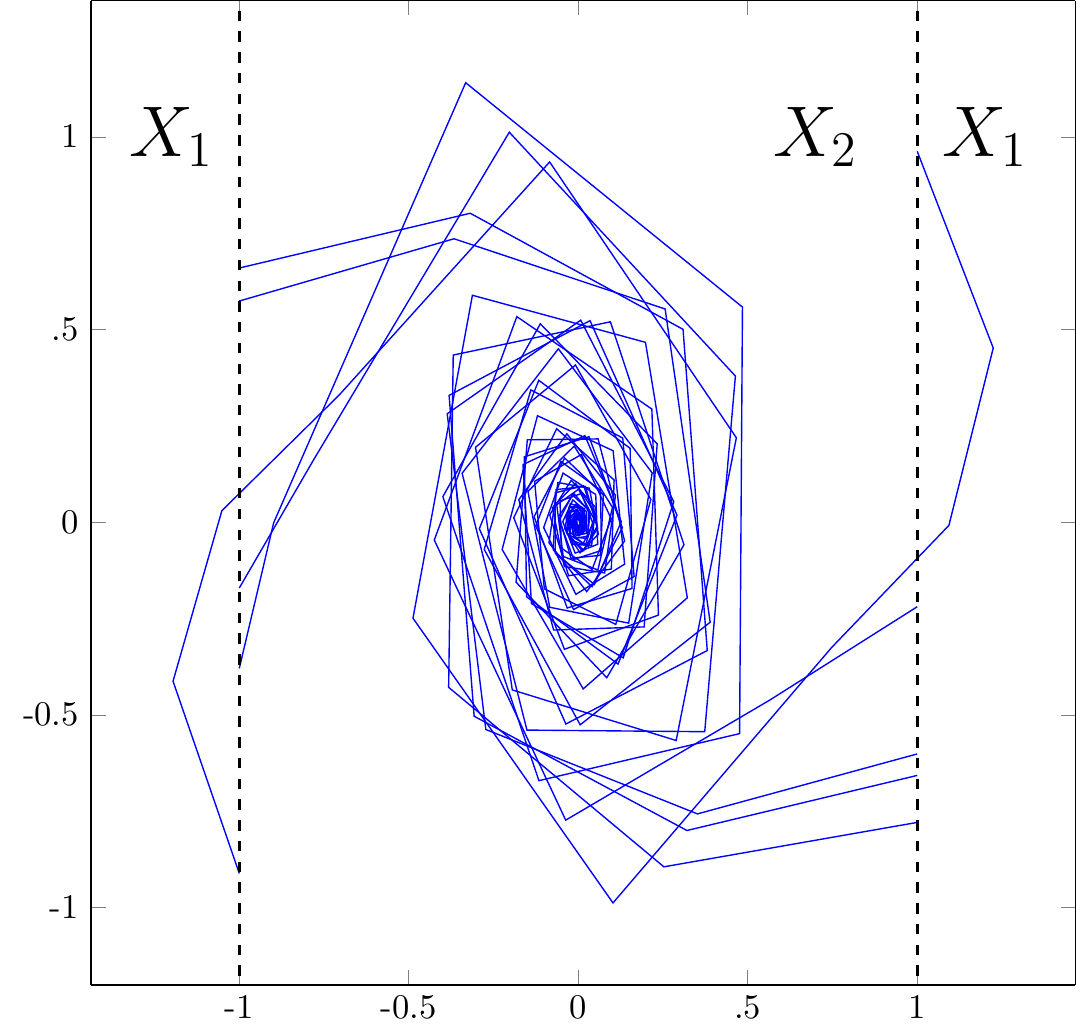}
    \caption{Running example simulation}
  \label{fig:running_ex_traces}
\end{figure}
\end{example}

We recall that our objective is to prove automatically that the set of the
possible trajectories of the system is bounded. This set, also called the reachable values set or the collecting
semantics of the system, is defined as follows:
\begin{equation}
\label{reachable}
\rea = \bigcup_{k \in \nn} T_{|_{\xlive}}^{(k)}(\xin)
\end{equation}
where $T_{|_{\xlive}}$ is the restriction of $T$ over $\xlive$.
The notation $T_{|_{\xlive}}^{(k)}$ stands for composing $k$-times the map $T_{|_{\xlive}}$.
\commentaire{
\ploccomment[inline]{Attention, ici $^k$ veut dire déroulé k fois, dans les
  pages précedentes, ça veut dire le $k$ieme element de la partition}
\ploccomment[inline]{C'est aussi implicitement le lift aux ensemble: $T(X) = \{
  T(x) | x \in X\}$ et donc $T_{|_{\xlive}} (X) = \{ ... \} \cap \xlive$}
 }
 
To prove this boundedness property, we can compute $\rea$ and do some analysis to prove that $\rea$ is
bounded. Nevertheless, the computation of $\rea$ cannot be done in general and instead, 
we have to compute an over-approximation of $\rea$ and show that this approximation is bounded. 

The usual abstract interpretation methodology to characterize and to construct this over-approximation
relies on the representation of $\rea$ as the smallest fixed point of a
monotone map over a complete lattice. 
In other words, $\rea$ satisfies:
\[
\rea=T(\rea\cap \xlive)\cup \xin=\bigcup_{i\in\ind} T^i(\rea\cap \xlive\cap X^i)\cup \xin \enspace.
\]
Let us define $\wp(\rd)$ the set of subsets of $\rd$ and introduce the function $F:\wp(\rd)\mapsto \wp(\rd)$ defined for all $C\in\wp(\rd)$ by:
\begin{equation}
\label{semantics}
F(C)=T(C\cap \xlive)\cup \xin=\bigcup_{i\in\ind} T^i(C\cap \xlive\cap X^i)\cup \xin \enspace.
\end{equation}
Thus, $\rea$ is the smallest fixed point of $F$ and from Tarski's Theorem, since
$F$ is monotone and $\wp(\rd)$ is a complete lattice:
\begin{equation}
\label{eq:fixpoint}
\rea=\min\{C\in\wp(\rd)\mid F(C)\subseteq C\}\enspace.
\end{equation}
Finally to compute an over-approximation of $\rea$ it suffices to compute a set $C$ such that $F(C)\subseteq C$.
A set $C$ which satisfies $F(C)\subseteq C$ is called an \emph{inductive invariant}\footnote{In the dynamical systems theory, the inductive invariant sets are called positive invariant.}.

The rest of the paper addresses the computation of a sound over-approximation of
$\rea$ using its definition as the smallest fixpoint of $F$ (Eq.~\eqref{eq:fixpoint}).



\section{Basic Semi-algebraic Inductive Invariants Set}
\label{sec:basicinductive-domainsSOS}
An easy way to over-approximate the set of reachable values is to restrict the set of inductive invariants that we consider.
We propose to restrict the class of such invariants to basic semi-algebraic sets using template abstractions.
A template abstraction consists of representing a given set as the intersection of sublevel sets of \emph{a-priori fixed}
functions depending on the state variables. Such functions are called \emph{templates}. Then computing an inductive invariant in the templates domain boils down to providing, for each template $p$, a bound $w(p)$ such that the intersection over the templates $p$ of sublevel sets $\{x\in\rd\mid p(x)\leq w(p)\}$ is an inductive invariant. In our context, a template is simply an {\em a-priori} fixed multivariate polynomial. 

The overall method is not new and corresponds to a specialization of the templates abstraction (see~\cite{DBLP:conf/esop/AdjeGG10,DBLP:journals/corr/abs-1111-5223})
to polynomial templates. However, in practice, the method developed in~\cite{DBLP:conf/esop/AdjeGG10,DBLP:journals/corr/abs-1111-5223} is restricted to template polynomials 
of degree 2 (quadratic forms) and affine systems or a very restricted class of piecewise affine systems. 

Next we give formal details about the polynomial template abstraction and the equations that must satisfy the template bounds vector $w$ to generate an inductive invariant.
From now on, we denote by $\pp$ the set of templates and by $\fr$ the set of functions from $\pp$ to $\br=\rr\cup\{-\infty,+\infty\}$. We equip $\fr$ with the functional 
partial order $\leq_\F$ i.e. $v\leq_\F w$ iff $v(p)\leq w(p)$ for all $p\in\pp$.  
Let $w\in\fr$. The sets that we consider are of the form:
\begin{equation}
\label{eq:basicinductive}
w^\star=\{x\in\rd\mid p(x)\leq w(p),\forall\, p\in\pp\} \,.
\end{equation}

\begin{example}
\label{semialgebraicrunning}
Let us define $q_1(x)=q_1(x_1,x_2)=x_1^2$, $q_2(x)=q_2(x_1,x_2)=x_2^2$ and let us consider a well-chosen polynomial $p$ of degree 6. We will explain in Subsection~\ref{initsub} how to generate automatically this template $p$. Let us define $\pprun:=\{q_1,q_2,p\}$.

Consider the function $w^0$ over $\pprun$, $w^0(q_1)=w^0(q_2)=2.1391$ and $w^0(p)=0$, the set ${w^0}^\star=\{(x_1,x_2)\in\rr^2\mid x_1^2\leq w^0(q_1),\ x_2^2\leq w^0(q_2),\ p(x)\leq w^0(p)\}$ 
is presented in Figure~\ref{invs}.

Now let us take the function $w^1$ over $\pprun$ defined by $w^1(q_1)=1.5503$,
$w^1(q_2)=1.9501$ and $w^1(p)=0$, the set ${w^1}^\star=\{(x_1,x_2)\in\rr^2\mid
x_1^2\leq w^1(q_1),\ x_2^2\leq w^1(q_2),\ p(x)\leq w^1(p)\}$ is also presented in Figure~\ref{invs}.

\begin{figure}[h]
  \begin{center}
    \includegraphics[width=.5\textwidth]{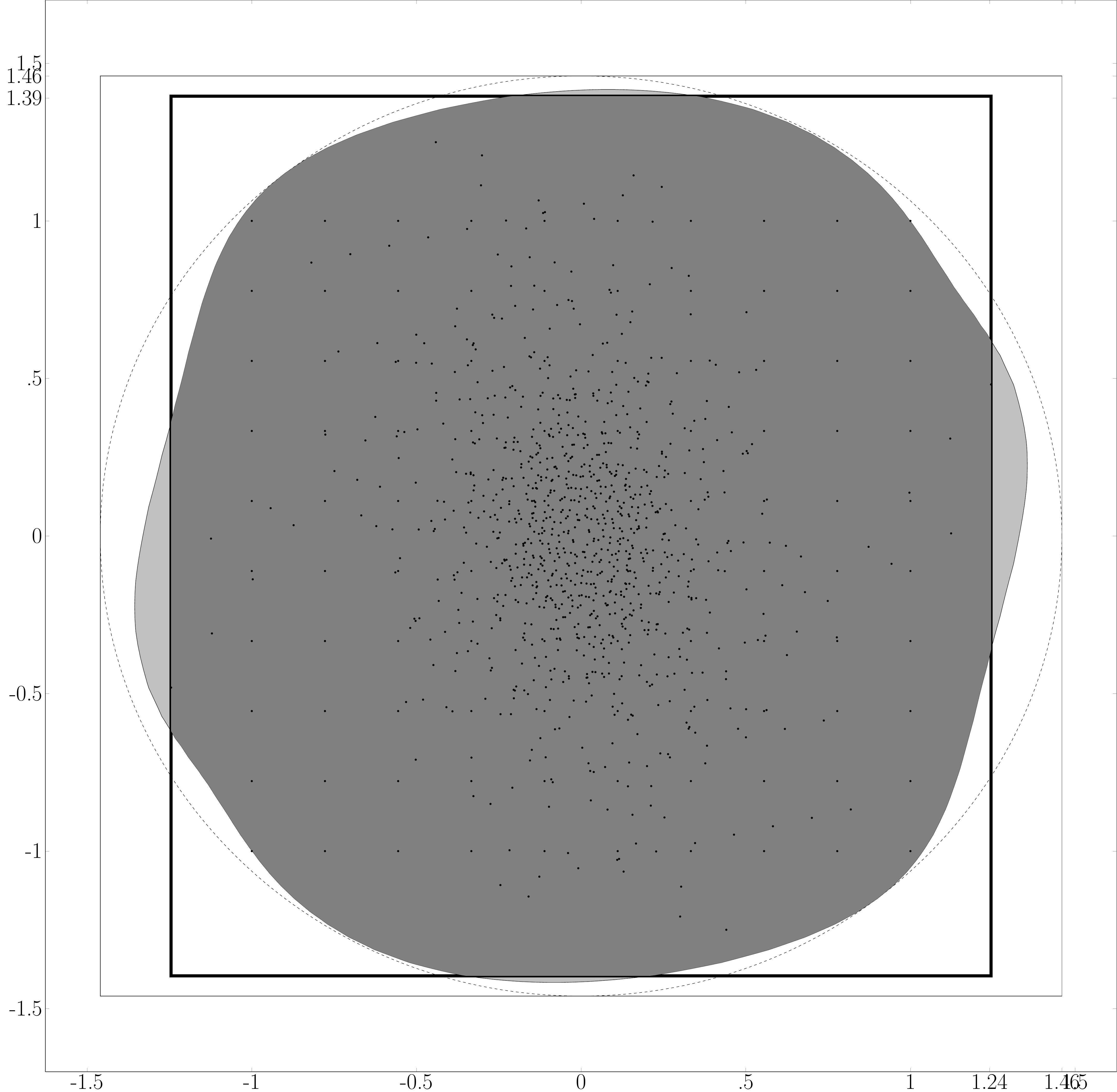}
  \end{center}

The semialgebraic set of gray denotes the set ${w^0}^\star$ of
Example~\ref{semialgebraicrunning} since $p(x) \leq 0$ is included in $x_i^2
\leq 2.1391$, for $i \in [1,2]$. The dark gray region denotes the semialgebraic
set ${w^0}^\star$. Black dots are actual reachable states of $\rea$ obtained by simulation.
\caption{Semialgebraic sets $w^\star$ for Example~\ref{semialgebraicrunning} }
\label{invs}
\end{figure}

Since the set of black dots in Figure~\ref{invs} belongs to ${w^0}^\star$ and ${w^1}^\star$, we guess that ${w^0}^\star$ and ${w^1}^\star$ contain both the reachability set of the system from Example~\ref{running}. 
To formally prove it, one way is to show that they are also inductive invariants for this system. 
\end{example}
We restrict the class of inductive invariants to those of the form~\eqref{eq:basicinductive} 
and characterize the inductiveness for such sets. Since 
each polynomial template $p$ is fixed, the considered variables that we handle are the template bounds
$w\in\fr$. Therefore, we need to translate the inductiveness of the sets $w^\star$ into inequalities 
on $w$. By definition the set $w^\star$ is an inductive invariant iff $F(w^\star)\subseteq w^\star$, that is:
\[
\bigcup_{i\in\ind} T^i(w^\star\cap \xlive\cap X^i)\cup \xin \subseteq w^\star\enspace .
\]
By definition, $w^\star$ is an inductive invariant iff:
\[
\forall\, p\in\pp,\ \forall x\in\bigcup_{i\in\ind} T^i(w^\star\cap \xlive\cap X^i)\cup \xin,\ p(x)\leq w(p) \enspace .
\]
Using the definition of the supremum, $w^\star$ is an inductive invariant iff:
\[
\forall\, p\in\pp,\ \sup_{\displaystyle{x\in\bigcup_{i\in\ind} T^i(w^\star\cap \xlive\cap X^i)\cup \xin}} p(x)\leq w(p) \enspace .
\]
Now, let us consider $p\in\pp$. Using the fact that for all $A,B\subseteq \rd$ and for all functions $f$, $\displaystyle{\sup_{A\cup B} f=\sup\{\sup_{A} f,\sup_B f\}}$:
\[
 \sup_{x\in\bigcup_{i\in\ind} T^i(w^\star\cap \xlive\cap X^i)\cup \xin}
 p(x)=\sup\left\{\sup_{i\in\ind}\; \sup_{x\in T^i(w^\star\cap \xlive\cap X^i)} p(x),\sup_{x\in \xin} p(x)\right\} \enspace .
\]
By definition of the image:
\[
 \sup_{x\in\bigcup_{i\in\ind} T^i(w^\star\cap \xlive\cap X^i)\cup \xin} p(x)=\sup\left\{\sup_{i\in\ind}\;\sup_{y\in w^\star\cap \xlive\cap X^i} p(T^i(y)),\sup_{x\in \xin} p(x)\right\} \enspace .
\]
Next, we introduce the following notation, for all $p\in\pp$:
\[
\sha{F_{i}}(w)(p) := \sup_{x\in w^\star\cap X^i\cap \xlive} p(T^i(x))\quad \text{ and } \quad {\xin}^\dagger(p) := \sup_{x\in \xin} p(x)\enspace.
\]
Finally, we define the function from $\fr$ to itself, for all $w\in\fr$: 
\[
\sha{F}(w) := \sup\left\{\sup_{i\in \ind} \sha{F_i}(w),{\xin}^\dagger\right\} \,.
\]
Note that $\sha{}$, ${}^\dagger$ correspond exactly to the notations used in~\cite{DBLP:conf/esop/AdjeGG10}.
By construction we obtain the following proposition:
\begin{proposition}
\label{boundineq}
Let $w\in\fr$. Then $w^\star$ is an inductive invariant (i.e. $F(w^\star)\subseteq w^\star$) iff $\sha{F}(w)\leq_\F w$.
\end{proposition}

From Prop.~\ref{boundineq}, $\inf\{w\in\fr^n\mid \sha{F}(w)\leq_\F w\}$ identifies the smallest inductive invariant $w^\star$ of the form~\eqref{eq:basicinductive}. 

\begin{example}
Let us consider the system defined at Example~\ref{running}. Let us consider the same templates basis $\pprun$ from Example~\ref{semialgebraicrunning} i.e. $\pprun=\{q_1,q_2,p\}$ where $q_1(x)=x_1^2$, $q_2(x)=x_2^2$ and $p$ is a well-chosen polynomial of degree 6. Let $w\in\frrun$. For $i=1$ and the templates $q_1$, we have:
\[
\sha{F_1}(w)(q_1)=\sup_{\substack{-x_1^2+1\leq 0\\ x_1^2\leq w(q_1),\ x_2^2\leq w(q_2),\ p(x)\leq w(p)}} (0.687x_1+0.558x_2-0.0001x_1 x_2)^2 \,.
\]
Indeed, $X^1=\{x\in\rr^2\mid -x_1^2+1\leq 0\}$ and $\xlive=\rr^2$ and the dynamics associated with $X^1$ is the polynomial function $T^1$ defined for all $x\in\rr^2$ by:
$T^1(x) =\left(\begin{smallmatrix} 
0.687x_1+0.558x_2-0.0001x_1 x_2 \\ 
-0.292x_1+0.773x_2
\end{smallmatrix}\right)$. Since $q_1$ computes the square of the first coordinates, this yields $q_1(T^1(x))=(0.687x_1+0.558x_2-0.0001x_1 x_2)^2$.

\end{example}

With $w\in\fr$, computing $\sha{F}(w)$ boils down to solving a finite number of nonconvex polynomial optimization problems. General methods do not exist to solve such problems. 
In Section~\ref{sec:relaxedSOS}, we propose a method based on Sums-of-Squares (SOS) to over-approximate $\sha{F}(w)$.


\section{SOS-based Relaxed Semantics }
\label{sec:relaxedSOS}
In this section, we introduce the relaxed functional on which we will compute a fixpoint, yielding a further over-approximation of the set $\rea$ of reachable values. This relaxed functional is constructed from 
a Lagrange relaxation of maximization problems involved in the evaluation of $\sha{F}$ and 
Sums-of-Squares strengthening of polynomial nonnegativity constraints. 
First, we recall mandatory background related to Sums-of-Squares and their application in polynomial optimization. 
The interested reader is referred to~\cite{lasserre2009moments} for more details.

\subsection{Sums-of-Squares Programming}
\label{sec:sos}

Let $\rr[x]_{2m}$ stands for the set of polynomials of degree at most $2 m$ and
$\Sigma[x] \subset \rr[x]$ be the cone of Sums-of-Squares (SOS) polynomials,
that is $\Sigma[x] := \{\,\sum_i q_i^2, \, \text{ with } q_i \in \rr[x] \,\}$. 

Our work will use the simple fact that 
for all $q\in\Sigma[x]$, then $q(x)\geq 0$ for all $x\in\rd$ as the set $\Sigma[x]$ contains 
only nonnegative polynomials. 
In other words, for any given polynomial $q$, we can strengthen the constraint of $q$ being nonnegative into the existence of an SOS decomposition of $q$.

%
For $q \in \rr[x]_{2m}$,  finding an SOS decomposition $\sum_i q_i^2 = q$ valid
over $\rd$ is equivalent to solve the following matrix linear feasibility problem:
\begin{align}
\label{eq:sdp}
q(x) = b_m(x)^T \, Q \, b_m(x) \,, \quad \forall x \in \rd, \,
\end{align}
where $b_m(x) := (1, x_1, \dots, x_d, x_1^2,x_1 x_2,\dots, x_d^m)$ (the vector
of all monomials in $x$ up to degree $m$) and  the Gram matrix $Q$,
being a {\em semidefinite positive} matrix (i.e.~all the eigenvalues of $Q$ are nonnegative). The size of $Q$ (as well as the length of $b_m$) is ${d + m \choose d}$. 

\begin{example}
Consider the bi-variate polynomial $q(x) := 1 + x_1^2 - 2 x_1 x_2 + x_2^2$. With $b_1 (x) = (1, x_1, x_2)$, one looks for a semidefinite positive matrix $Q$ such that the polynomial equality $q(x) = b_1(x)^T \, Q \, b_1(x)$ holds for all $x \in \rr^2$. 
The matrix \[Q = 
\begin{pmatrix}
1 & 0 & 0\\
0 & 1 & -1\\
0 & -1 & 1
\end{pmatrix}\] satisfies this equality and has three nonnegative eigenvalues,
which are 0, 1, and 2, respectively associated to the three eigenvectors $e_0 :=
(0, 1/\sqrt{2}, 1/\sqrt{2})^\intercal $, $e_1 := (1, 0, 0)^\intercal $ and $e_2 := (0, 1/\sqrt{2}, -1/\sqrt{2})^\intercal$. Defining the matrices $L
:= (e_1 \, e_2 \, e_0)=\begin{psmallmatrix}1 & 0 & 0\\ 0 & \frac{1}{\sqrt{2}} & \frac{1}{\sqrt{2}} \\ 0 & -\frac{1}{\sqrt{2}} & \frac{1}{\sqrt{2}} \end{psmallmatrix}$ and $D = \begin{psmallmatrix} 1 & 0 & 0 \\ 0 & 2 & 0 \\ 0 & 0 & 0 \end{psmallmatrix}$, one obtains the decomposition $Q = L^\intercal \, D \, L$ and the equality $q(x) = (L \,
b_1(x))^T \, D \, (L \, b_1(x)) = \sigma (x) = 1 + (x_1 - x_2)^2$, for all $x
\in \rr^2$. The polynomial $\sigma$ is called an {\em SOS certificate} and
guarantees that $q$ is nonnegative.
\end{example}
In practice, one can solve the general problem~\eqref{eq:sdp} by using
semidefinite programming (SDP) solvers (e.g. {\sc Mosek}~\cite{mosek},
SDPA~\cite{Yamashita10SDPA}, CSDP~\cite{csdp99borchers}). For more details about SDP, we refer the interested reader to~\cite{Vandenberghe94semidefiniteprogramming}.

The SOS reinforcement of polynomial optimization problems consists of restricting polynomial nonnegativity to being an element of $\Sigma[x]$. In case of polynomial maximization problems, the SOS reinforcement boils down to computing an upper bound 
of the real optimal value. For example let $p \in \rr[x]$ and consider the unconstrained polynomial maximization problem $\sup\{p(x),\ x\in\rd\}$. 
Applying SOS reinforcement, we obtain:
\begin{equation}
\label{sosapplicationunconstrained}
\sup\{p(x),\ x\in\rd\}=\inf\{\eta\mid \eta-p(x)\geq 0\}\leq \inf\{\eta\mid \eta-p(x)\in\Sigma[x]\} \enspace .
\end{equation}
Now, let $p,q\in \rr[x]$ and consider the constrained polynomial maximization problem: $\sup \{  p(x) | q(x) \leq 0,\ x\in\rd\}$.
Let $\lambda\in\Sigma[x]$, then:
\[
  \sup_{
      q(x) \leq 0,\
      x\in\rd
    } \quad
  p(x)\leq 
  \sup_{%
    \substack{%
      x \in \rd}} \quad
  p(x) - \lambda(x) \cdot q(x)\enspace .
\] 
Indeed, suppose $q(x)\leq 0$, then $-\lambda(x)q(x)\geq 0$ and $p(x)\leq p(x)-\lambda(x)q(x)$.
Finally taking the supremum over $\{x\in\rd\mid q(x)\leq 0\}$ provides the above inequality.
Since $\sup\{p(x) - \lambda(x) \cdot q(x),\ x\in\rd\}$ is an unconstrained polynomial maximization problem then we apply 
an SOS reinforcement (as in Eq.~\eqref{sosapplicationunconstrained}) and we obtain:  
\[
\sup_{
      q(x) \leq 0,\ x\in\rd
    } \quad
  p(x)\leq 
  \sup_{%
      x \in \rd} \quad
  p(x) - \lambda(x) \cdot q(x) 
  \leq \inf\{\eta\mid \eta-p-\lambda q\in \Sigma[x]\} \,.
  \]
  Finally, note that this latter inequality is valid whatever $\lambda\in\Sigma[x]$ and so we can take the infimum 
  over $\lambda\in\Sigma[x]$ which leads to:
  \begin{equation}
  \label{SOSrelaxationconstrained}
\sup_{
      q(x) \leq 0,\ x\in\rd
    } \quad
  p(x)\leq \inf_{\lambda\in\Sigma[x]}
  \sup_{%
      x \in \rd} \quad
  p(x) - \lambda(x) \cdot q(x) 
  \leq \inf_{\substack{\eta-p-\lambda q\in \Sigma[x]\\ \lambda\in\Sigma[x]}} \quad \eta \,.
  \end{equation}
  
  In Eq.~\eqref{SOSrelaxationconstrained}, $\lambda$ is an SOS polynomial but to exploit linear programming solvers in policy iterations (see the fourth assertion of Prop.~\ref{phiproperty}) 
  we restrict $\lambda$ to be a nonnegative scalar and in this case, since positive scalars are sum-of-squares polynomials of degree 0, we obtain a
  safe over-approximation of the right-hand-side of Eq.~\eqref{SOSrelaxationconstrained}.

  In presence of several constraints, we assign to each constraint an element $\sigma\in\Sigma[x]$, and we consider 
  the product of $\sigma$ with its associated constraint and then the sum of all such products. This sum is finally added to the objective 
  function. 
  
The use of such SOS polynomials for constrained polynomial optimization
problem can be seen as a generalization of the S-procedure from~\cite{Yakubovich71}. We refer to~\cite{Las01moments} or~\cite{Parrilo2003relax} for applications in control. Note that the existence of SOS decompositions of positive polynomials over compact sets is ensured by the Putinar Positivstellensatz from~\cite{Putinar1993positive}.

\subsection{Relaxed semantics}
The computation of $\sha{F}$ as a polynomial maximization problem cannot be 
directly performed using numerical solvers. We use the SOS reinforcement mechanisms
described above to relax the computation and characterize an abstraction of
$\sha{F}$. 

We still assume the knowledge of the template basis $\pp$, involving polynomials of degree at most $2 m$.
Let us define $\frp$ the set of nonnegative functions over $\pp$ i.e. $g\in\frp$ iff for all $p\in\pp$, $g(p)\in\rr_+$.
Let $p\in\pp$ and $w\in\fr$. Starting from the definition of $\sha{F_i}$, one obtains the following:
\allowdisplaybreaks[4]
\[
\begingroup
\begin{array}[t]{lcl}
\mybrackets{\sha{F_i}(w)}(p) &=&
\displaystyle\sup_{\substack{q(x)\leq w(q),\ \forall q\in\pp\\ r_j^i(x) \leq 0,\ \forall\, j\in\iset{n_i}\\ r_k^0(x)\leq 0,\ 
\forall\, k\in\iset{n_0}}} \quad p(T^i(x))\\
&\leq&
\begin{array}[t]{lll}
\displaystyle{\inf_{\substack{\lambda\in\frp\\ \sigma \in\Sigma[x], \mu_l\in\Sigma[x],\gamma_l\in\Sigma[x] \\ \deg (\sigma) \leq 2 m \deg T^i \\ \deg (\mu_l r_l^i) \leq 2 m \deg T^i \\ \deg (\gamma_l r_l^0) \leq 2 m \deg T^i }}}&\displaystyle{\sup_{x\in \rd}} &\displaystyle{ p(T^i(x))+\sum_{\mathrm{q \in \pp}}\lambda(q)
(w(q) - q(x))} \\[-4em]
&&\displaystyle{- \sum_{\mathrm{l=1}}^{\mathrm{n^i}}\mu_l(x) r_l^i(x) -  
\sum_{\mathrm{l=1}}^{\mathrm{n^0}}\gamma_l(x) r_l^0(x)}
\end{array}\vspace{1.5em}
\\
%
%
&\leq& 
\begin{array}[t]{cl}\displaystyle{\inf_{\lambda,\sigma,\mu_l,\gamma_l,\eta}} &\eta\\
\st&
\left\{
\begin{array}{l}
\displaystyle{\eta- p\circ T^i-\sum_{\mathrm{q \in \pp}}\lambda(q) (w(q) - q)+ \sum_{\mathrm{l=1}}^{\mathrm{n^i}}\mu_l r_l^i +  
\sum_{\mathrm{l=1}}^{\mathrm{n^0}}\gamma_l r_l^0} = \sigma \,,\\
\lambda\in\frp,\ \sigma\in\Sigma[x],\ \mu_l\in\Sigma[x],\ \gamma_l\in\Sigma[x],\ \eta\in \rr \,, \\
\deg (\sigma) \leq 2 m \deg T^i \,, \\
\deg (\mu_l r_l^i) \leq 2 m \deg T^i ,\ \deg (\gamma_l r_l^0) \leq 2 m \deg T^i \,.
\end{array}
\right.
\end{array}\\
\multicolumn{3}{r}{\textrm{(using an SOS reinforcement to remove the sup)}}\\
\end{array}
\endgroup
\]
We denote by $\Sigma[x]^n$ the set of $n$-tuples of SOS polynomials.
For clarity purpose, the dependency on $i$ is omitted within the notations of the multipliers $\mu_l$ and $\gamma_l$.
Moreover, let us write $\sum_{l=1}^{n^i} \mu_l r_l^i$ (resp. $\sum_{\mathrm{l=1}}^{\mathrm{n^0}}\gamma_l r_l^0$) as $\langle\mu, r^i\rangle$ (resp. $\langle\gamma, r^0\rangle$). Finally, we write $\mybrackets{\becarre{F_i}(w)}(p)$ the over-approximation of $\mybrackets{\sha{F_i}(w)}(p)$, defined as follows:
\begin{equation}
\label{relaxedpoly}
\begin{array}{ll}
 \mybrackets{\becarre{F_i}(w)}(p) =  &\displaystyle{\inf_{\lambda,\sigma, \mu,\gamma,\eta}} \eta\\
&\st
\left\{
\begin{array}{l}
\displaystyle{\eta- p\circ T^i-\sum_{\mathrm{q \in \pp}}\lambda(q) (w(q) - q) + \langle\mu, r^i\rangle + \langle\gamma, r^0\rangle} = \sigma\\
\lambda\in\frp,\ \sigma\in\Sigma[x] ,\ \mu\in\Sigma[x]^{n_i},\ \gamma\in\Sigma[x]^{n_0},\ \eta\in \rr \,,\\
\deg (\sigma) \leq 2 m \deg T^i ,\ \deg (\langle\mu, r^i\rangle + \langle\gamma, r^0\rangle) \leq 2 m \deg T^i \,.
\end{array}
\right.
\end{array}
\end{equation}
In Equation~\eqref{relaxedpoly}, the notation $\lambda$ is a vector of Lagrange multipliers. Each multiplier is associated with a constraint constructed from a template i.e. a constraint $q(x)-w(q)\leq 0$. We also introduce the vector of SOS polynomials $\mu$ and $\gamma$. Their role is to take into account the presence of the constraints $x\in X^i$ and $x\in \xlive$ in the computation of $ \mybrackets{\becarre{F_i}(w)}(p)$. Recall that $X^i$ and $\xlive$ are basic semi-algebraic sets, then the size of the vectors $\mu$ 
and $\gamma$ are equal to the number of polynomials defining $X^i$ and $\xlive$. 

We conclude that, for all $i\in\ind$, the evaluation of $\becarre{F_i}$ can be done using SOS programming, since it 
is reduced to solve a minimization problem with a linear objective function and
linear combination of polynomials constrained to be sum-of-squares.

Note that $\rel{F_i}$ defined at Eq.~\eqref{relaxedpoly} is the SOS extension of the relaxed function defined in~\cite{DBLP:journals/corr/abs-1111-5223}. Indeed, considering the special case where $T^i$ is affine, the templates $p,q$ and the test functions $r^i$, $r^0$ are quadratic, the vectors 
$\mu^i$ and $\gamma^i$ are restricted to be nonnegative scalars, 
then $\rel{F_i}$ corresponds to the relaxed function defined in~\cite{DBLP:journals/corr/abs-1111-5223} at Eq.~(3.12).
%
\begin{example}
\label{examplerunning2}
We still consider the running example defined at Example~\ref{running} and take again the same templates basis $\pprun$ of Example~\ref{semialgebraicrunning} composed of $q_1:x\mapsto x_1^2$ and $q_2:x\mapsto x_2^2$. and a well-chosen polynomial $p$ of degree 6. For the index of the partition $i=1$.
Recall that $T^1(x) =\left(\begin{smallmatrix} 
0.687x_1+0.558x_2-0.0001x_1 x_2 \\ 
-0.292x_1+0.773x_2
\end{smallmatrix}\right)$ and  $X^1=\{x\in\rr^2\mid -x_1^2+1\leq 0\}$ and thus $r_1^1(x)=-x_1^2+1$. Let $w\in \frrun$,
then:
\[ 
\begin{array}{l}
\mybrackets{\becarre{F_1}(w)}(q_1) =\\
\displaystyle{\inf_{\lambda,\sigma, \mu,\eta}} \eta\\
\st
\left\{
\begin{array}{l}
\eta- (0.687x_1+0.558x_2-0.0001x_1x_2)^2-\lambda(q_1) (w(q_1) - x_1^2)\\
-\lambda(q_2) (w(q_2) - x_2^2)-\lambda(p) (w(p) - p(x)) +\mu(x)(1-x_1^2) = \sigma(x)\\
\lambda\in\frp,\ \sigma\in\Sigma[x] ,\ \mu\in\Sigma[x],\ \eta\in \rr \,,\\
\deg (\sigma) \leq 6,\ \deg (\mu) \leq 6\,.
\end{array}
\right.
\end{array}
 \]
In practice, one cannot find any feasible solution of degree less than 6, thus we replace the degree constraint by the more restrictive one: $\deg (\sigma) \leq 6,\ \deg (\mu) \leq 6$.  
\end{example}
The computation of $\sha{F}$ requires the approximation of ${\xin}^\dag := \sup\{p(x), x\in\xin\}$.
Since $\xin$ is a basic semi-algebraic set and each template $p$ is a polynomial, then the evaluation of ${\xin}^\dag$ boils down to
 solving a polynomial maximization problem. Next, we use SOS reinforcement described above to over-approximate  ${\xin}^\dag$ with
the set $\becarre{{\xin}}$, defined as follows:
\[
\becarre{{\xin}}(p):=\inf\left\{\eta\left| 
\begin{array}{c}
\eta-p+\langle\nu^{n_\mathrm{in}}, r^{n_\mathrm{in}}\rangle = \sigma_0,\\
\eta\in\rr,\ \sigma_0\in\Sigma[x], \nu^{\mathrm{in}}\in\Sigma[x]^{n_\mathrm{in}},\\
\deg (\sigma_0) \leq 2 m, \deg ( \langle\nu^{n_\mathrm{in}}, r^{n_\mathrm{in}}\rangle ) \leq 2 m 
\end{array}
\right\}\right.
 \enspace.
\]
Thus, the value of $\becarre{{\xin}}(p)$ is obtained by solving an SOS optimization problem. Since $\xin$ is a nonempty compact basic semi-algebraic set, this problem has a feasible solution (see the proof of ~\cite[Th. 4.2]{Las01moments}), ensuring that $\becarre{{\xin}}(p)$ is finite valued.

\begin{example}
The initialization set $\xin$ of Example~\ref{running} is $[-1,1]\times [ -1,1]$. It can be written as:
$\{(x_1,x_2)\in\rr^2\mid x_1^2-1\leq 0,\ x_2^2-1\leq 0\}$. 
Then, considering the same template basis of Example~\ref{examplerunning2} and the template $q_1$:
\[
 \becarre{{\xin}}(q_1):=\inf\left\{\eta\left| 
\begin{array}{c}
\eta-x_1^2+\nu_1^{n_\mathrm{in}}(x)(x_1^2-1)+\nu_2^{n_\mathrm{in}}(x)(x_2^2-1) = \sigma_0(x),\\
\eta\in\rr,\ \sigma_0\in\Sigma[x], \nu_1^{\mathrm{in}},\nu_2^{\mathrm{in}}\in\Sigma[x],\\
\deg (\sigma_0) \leq 6, \deg ( \langle\nu_1^{n_\mathrm{in}}) \leq 6, \deg ( \langle\nu_2^{n_\mathrm{in}}) \leq 6
\end{array}
\right\}\right.
 \enspace.
\]
It is easy to see that taking for all $x\in\rr^2$, $\nu_1^{n_\mathrm{in}}(x)=1$ and for all $x\in\rr^2$, $\nu_2^{n_\mathrm{in}}(x)=0$ leads to
$\eta-x_1^2+\nu_1^{n_\mathrm{in}}(x)(x_1^2-1)+\nu_2^{n_\mathrm{in}}(x)(x_2^2-1) = \eta-1=\sigma_0(x)$. Thus for $\eta=1$ and for all $x\in\rr^2$, 
$\sigma_0(x)=0$, we obtain $\becarre{{\xin}}(q_1)\leq 1$. We will see at Prop.~\ref{safety}, that ${\xin}^\dag\leq_\F \becarre{{\xin}}$. Thus, since ${\xin}^\dag(q_1)=\sup\{x_1^2\mid (x_1,x_2)\in [-1,1]\times [ -1,1]\}=1$, we conclude that $1\leq \becarre{{\xin}}(q_1)$ and $\becarre{{\xin}}(q_1)=1$.
\end{example}

Finally, we define the relaxed functional $\becarre{F}$ for all $w\in\fr$ for all $p\in\pp$ as follows:
\begin{equation}
\label{relaxedfunctional}
  \mybrackets{\becarre{F}(w)}(p) = 
  \displaystyle{\sup\left\{\sup_{i\in\ind}
      \mybrackets{\becarre{F_i}(w)}(p),\becarre{{\xin}}(p)\right\}} \enspace.
\end{equation}
As we followed the construction proposed in Section~\ref{sec:sos}, the relaxed functional $\becarre{F}$ provides a safe over-approximation of the abstract semantics $\sha{F}$.

\begin{proposition}[Safety]
\label{safety}
The following statements hold:
\begin{enumerate}
\item ${\xin}^\dag\leq_\F \becarre{{\xin}}$; 
\item For all $i\in\ind$, for all $w\in\fr$, $\sha{F_i}(w) \leq_\F \becarre{F_i}(w)$;
\item For all $w\in \fr$, $\sha{F}(w) \leq_\F \becarre{F}(w)$.
\end{enumerate}
\end{proposition}

An important property that we will use to prove some results on policy iteration algorithm is the monotonicity 
of the relaxed functional.
\begin{proposition}[Monotonicity]
\label{monotonicity}
\begin{enumerate}
\item For all $i\in\ind$, $w\mapsto \becarre{F_i}(w)$ is monotone on $\fr$;
\item The function $w\mapsto \becarre{F}(w)$ is monotone on $\fr$.
\end{enumerate}
\end{proposition}
\begin{proof}
  
  Let $p\in\pp$. For $v\in\fr$, $\lambda\in\frp$, $\mu\in\Sigma[x]^{n_i}$, $\gamma\in\Sigma[x]^{n_0}$ and $\eta\in\rr$ such that $\deg (\langle\mu, r^i\rangle + \langle\gamma, r^0\rangle) \leq 2 m \deg T^i$, we define the polynomial in $x$, $\psi^{\lambda,\mu,\gamma,\eta}(v):=\eta - p\circ T^i-\sum_{\mathrm{q \in
      \pp}}\lambda(q) (v(q) - q) + \langle\mu, r^i\rangle + \langle\gamma,
  r^0\rangle$. We define for $v\in\fr$ the set $R(v)=\{(\lambda,\mu,\gamma,\eta)\mid \psi^{\lambda,\mu,\gamma,\eta}(v)\in \Sigma[x],\ \lambda\in\frp, \mu\in\Sigma[x]^{n_i}, \gamma\in\Sigma[x]^{n_0}, \eta\in\rr\}$.
  Now, let us take $w,w'\in \fr$ such that $w \leq_\F w'$. We have: 
  \[
  \begin{array}{ll}
  &\psi^{\lambda,\mu,\gamma,\eta}(w)\\
  =&\eta - p\circ T^i-\sum_{\mathrm{q \in\pp}}\lambda(q) (w(q) - q) + \langle\mu, r^i\rangle + \langle\gamma,r^0\rangle\\
  =&\eta - p\circ T^i-\sum_{\mathrm{q \in \pp}}\lambda(q) (w(q)-w'(q)+w'(q) - q) + \langle\mu, r^i\rangle + \langle\gamma,r^0\rangle\\
  =&\eta - p\circ T^i-\sum_{\mathrm{q \in \pp}}\lambda(q) (w'(q) - q) + \langle\mu, r^i\rangle + \langle\gamma,r^0\rangle
  -\sum_{\mathrm{q \in \pp}}\lambda(q) (w(q)-w'(q))\\
  =&\psi^{\lambda,\mu,\gamma,\eta}(w') +\sum_{\mathrm{q \in \pp}}\lambda(q) (w'(q) - w(q)) \,.
  \end{array}
  \]
   Then, from $w \leq_\F w'$ and the fact that $\lambda(q)$ are nonnegative scalars, if $\psi^{\lambda,\mu,\gamma,\eta}(w')$ is an SOS polynomial, so is $\psi^{\lambda,\mu,\gamma,\eta}(w)$ as a sum of a SOS polynomial and a nonnegative scalar. Hence, we have
   $R(w')\subseteq R(w)$. Finally, we recall that if $A \subseteq B$, then $\inf_B \leq \inf_A$. We conclude that $\mybrackets{\becarre{F_i}(w)}(p) \leq \mybrackets{\becarre{F_i}(w')}(p)$.
   
  2. The mapping $\becarre{F}$ is monotone as supremum of monotone maps.
\end{proof}

From the third assertion of Prop.~\ref{safety}, if $w$ satisfies $\rel{F}(w)\leq_\F w$ then 
$\sha{F}(w)\leq_\F w$ and from Prop.~\ref{boundineq}, $w^\star$ is an inductive invariant 
and thus $\rea\subseteq w^\star$. This result is formulated as the following corollary.
\begin{corollary}[Over-approximation]
For all $w\in\fr$ such that $\becarre{F}(w)\leq_\F w$ then $\rea\subseteq w^\star$.
\end{corollary}

\section{Policy Iteration in Polynomial Templates Abstract Domains}
\label{sec:policySOS}

We are interested in computing the least fixpoint $\becarre{\rea}$ of
$\becarre{F}$, $\becarre{\rea}$ being an over-approximation of $\rea$ (least fixpoint of
  $F$). As for the definition of $\rea$, it can be reformulated
  using Tarski's theorem as the minimal post-fixpoint:
\[
\becarre{\rea} = \min \{ w \in \fr | \becarre{F}(w) \leq_\F w \} \,.
\]
The idea behind policy iteration is to over-approximate $\becarre{\rea}$ using successive iterations which are composed of
\begin{itemize}
\item the computation of polynomial template bounds using linear programming,
\item the determination of new policies using SOS programming,
\end{itemize}
until a fixpoint is reached. 
Policy iteration navigates in the set of post-fixpoints of $\becarre{F}$ and needs to start from a post-fixpoint $w^0$
known {\em a-priori}. It acts like a narrowing operator and can be interrupted at any time. For further information on policy iteration,
the interested reader can consult~\cite{costan2005policy,DBLP:conf/esop/GaubertGTZ07}.

\subsection{Policies}
Policy iteration can be used to compute a fixpoint of a monotone self-map defined as an infimum of a family of 
affine monotone self-maps. In this paper, we propose to design a policy iteration algorithm to compute a fixpoint of $\becarre{F}$.
In this subsection, we give the formal definition of policies in the context of polynomial templates and define the family
of affine monotone self-maps. We do not apply the concept of policies on $\becarre{F}$ but on the functions $\becarre{F_i}$
exploiting the fact that for all $i\in\ind$, $\becarre{F_i}$ is the optimal value of a minimization problem.

Policy iteration needs a \emph{selection property}, that is, when an element $w\in\fr$ is given, there exists a policy which achieves the 
infimum. In our context, since we apply the concept of policies to $\becarre{F_i}$, it means that the minimization 
problem involved in the computation of $\becarre{F_i}$ has an optimal solution. In our case, for $w\in\fr$ and $p\in\pp$, an optimal solution is a vector $(\lambda,\sigma,\mu, \gamma)\in \frp\times\Sigma[x]\times
\Sigma[x]^{n_i} \times \Sigma[x]^{n_0}$ such that, using
  (\ref{relaxedpoly}), we obtain:
\begin{equation}
\label{optimalsolution}
\begin{array}{c}
\displaystyle{
\mybrackets{\becarre{F_i}(w)}(p)=p\circ T^i+\sum_{q \in \pp}\lambda(q) (w(q) - q) -\langle\mu, r^i\rangle -\langle\gamma, r^0\rangle +\sigma}\\
\text{and }\deg (\sigma) \leq 2 m \deg T^i ,\ \deg (\langle\mu, r^i\rangle + \langle\gamma, r^0\rangle) \leq 2 m \deg T^i
\end{array}
\enspace .
\end{equation}
Observe that in Eq.~\eqref{optimalsolution}, $\mybrackets{\becarre{F_i}(w)}(p)$ is a scalar whereas
the right-hand-side is a polynomial. The equality in this equation means that this polynomial is a constant polynomial. 
Then we introduce the set of feasible solutions for the SOS problem $\mybrackets{\becarre{F_i}(w)}(p)$:
\begin{equation}
\label{soldef}
\Sol(w,i,p)=\{(\lambda,\sigma, \mu, \gamma)\in \frp\times \Sigma[x]\times \Sigma[x]^{n_i} \times \Sigma[x]^{n_0}\mid \text{Eq.~\eqref{optimalsolution} holds}\}\enspace.
\end{equation}
Since policy iteration algorithm can be stopped at any step and still provides a sound over-approximation, we stop the iteration when $\Sol(w,i,p)=\emptyset$.
Now, we are interested in the elements $w\in\frr$ such that $\Sol(w,i,p)$ is non-empty:
\begin{equation}
\label{slaterspace}
\FS=\{w\in \fr\mid \forall\, i\in\ind,\ \forall\, p\in\pp,\ \Sol(w,i,p)\neq \emptyset\}\enspace.
\end{equation}
The notation $\FS$ was introduced in~\cite{DBLP:journals/corr/abs-1111-5223} to define the elements $w\in\fr$ satisfying 
 $\Sol(w,i,p)\neq \emptyset$. In~\cite[Section 4.3]{DBLP:journals/corr/abs-1111-5223}, we
 could ensure that $\Sol(w,i,p)\neq \emptyset$ using Slater's 
 constraint qualification condition. 
 In the current nonlinear setting, we cannot use the same condition, which yields a more complicated definition for $\FS$.
 
 Finally, we can define a policy as a map which selects, for all $w\in\FS$, for all $i\in\ind$ and for all $p\in\pp$ a vector of $\Sol(w,i,p)$.
More formally, we have the following definition:
\begin{definition}[Policies in the policy iteration SOS based setting]
A policy is a map $\pi:\FS\mapsto ((\ind\times \pp)\mapsto \frp\times\Sigma[x]\times\Sigma[x]^{n_i}\times\Sigma[x]^{n_0})$ such that:
 $\forall\, w\in\FS$, $\forall\, i\in\ind$, $\forall\, p\in\pp$, $\pi(w)(i,p)\in \Sol(w,i,p)$.   
\end{definition}
We denote by $\Pi$ the set of policies. For $\pi\in\Pi$, let us define $\pi_\lambda$ as the map from $\FS$ to $(\ind\times \pp)\mapsto \frp$ which associates with $w\in\FS$ and $(i,p)\in\ind\times\pp$ the first element of $\pi(w)(i,p)$ i.e. if $\pi(w)(i,p)=(\lambda,\sigma,\mu,\gamma)$ then $\pi_\lambda(w)(i,p)=\lambda$. The equality  $\pi_\lambda(w)(i,p)=\lambda$ means that when we perform the policy iterations algorithm, we select the vector of Lagrange multipliers $\lambda$ associated with the constraints of the form $q(x)\leq w(q)$. The purpose of this selection is to update the value of $w$ using the direction $\lambda$. The other coordinates composing $\pi(w)(i,p)$ that is $\sigma,\mu,\gamma$ do not serve the policy iterations algorithm but are only used to take in consideration the sets $X^i$ and $\xlive$ in the computation of $\becarre{F_i}(w)(p)$.

As said before, policy iteration exploits the linearity of maps when a policy is fixed. We have to define the affine maps we will use in a policy iteration step. With $\pi\in\Pi$, $w\in\FS$, $i\in\ind$ and $p\in\pp$ and $\lambda=\pi_\lambda(w)(i,p)$, let us define the map $\phi_{w,i,p}^{\lambda}:\fr\mapsto\br$ as follows:
\begin{equation}
\label{auxiliary}
v\mapsto \phi_{w,i,p}^{\lambda}(v)=\sum_{q \in \pp}\lambda(q) v(q)+\mybrackets{\becarre{F_i}(w)}(p)-\sum_{q \in \pp}\lambda(q) w(q) \,.
\end{equation}
Then, for $\pi\in\Pi$, we define for all $w\in\FS$, the map $\Phi_w^{\pi(w)}$ from $\fr\mapsto \fr$. Let $v\in\fr$ and $p\in\pp$:
\begin{equation}
\label{auxiliaryfixpoint}
\Phi_w^{\pi(w)}(v)(p)=\sup\left\{\sup_{i\in\ind}\phi_{w,i,p}^{\lambda}(v),\becarre{{\xin}}(p)\right\} \,.
\end{equation}

\begin{example}
Let us consider Example~\ref{examplerunning2} and the function $w^0(q_1)=w^0(q_2)=2.1391$ and $w^0(p)=0$. 
Then there exists two SOS polynomials $\overline{\mu}$ and $\overline{\sigma}$ such that, for all $x\in \rd$:
\[ 
\begin{array}{rr}
\mybrackets{\becarre{F_1}(w)}(q_1) =&(0.687x_1+0.558x_2-0.0001x_1x_2)^2+\overline{\lambda}(q_1) (2.1391 - x_1^2)\\
&+\overline{\lambda}(q_2) (2.1391 - x_2^2)-\overline{\lambda}(p)p(x) -\overline{\mu}(x)(1-x_1^2) +\overline{\sigma}(x)\\
=&1.5503 \,,
\end{array}
 \]
 with $\overline{\lambda}(q_1)=\overline{\lambda}(q_2)=0$ and $\overline{\lambda}(p)=2.0331$. It means that $\overline{\lambda}$, $\overline{\mu}$ and 
 $\overline{\sigma}$ are computed such that $(0.687x_1+0.558x_2-0.0001x_1x_2)^2+\overline{\lambda}(q_1) (2.1391 - x_1^2)+\overline{\lambda}(q_2) (2.1391 - x_2^2)-\overline{\lambda}(p)p(x) -\overline{\mu}(x)(1-x_1^2) +\overline{\sigma}(x)$ is actually 
a constant polynomial.

 Then $(\overline{\lambda},\overline{\mu},\overline{\sigma})\in\Sol(w^0,1,q_1)$
 and we can define a policy $\pi(w^0)$ such that $\pi(w^0)(1,q_1)=(\overline{\lambda},\overline{\mu},\overline{\sigma})$ and thus 
 $\pi_\lambda(w^0)(1,q_1)=(0,0,2.0331)$. We can thus define for $v\in\frrunr$, the affine mapping: $\phi_{w^0,1,q_1}^{\lambda}(v)=\lambda(q_1) v(q_1)+\lambda(q_2)v(q_2)+\lambda(p)v(p)+\mybrackets{\becarre{F_1}(w)}(q_1)-\lambda(q_1) w(q_1)-\lambda(q_2)w(q_2)-\lambda(p)w(p)=2.1391 v(p)+1.5503$.
\end{example}

Let us denote by $\frr$ the set of finite valued function on $\pp$ i.e $g\in\frr$ iff $g(p)\in\rr$ for all $p\in\pp$.
\begin{proposition}[Properties of $\phi_{i,w,p}^{\lambda}$]
\label{phiauxproperty}
Let $\pi\in\Pi$, $w\in\FS$ and $(i,p)\in\ind\times \pp$. Let us write $\lambda=\pi_\lambda(w)(i,p)$. The following properties are true:
\begin{enumerate}
\item $\phi_{w,i,p}^{\lambda}$ is affine on $\frr$\,;
\item $\phi_{w,i,p}^{\lambda}$ is monotone on $\fr$\,;
\item $\forall\, v\in \fr$, $\becarre{F_i}(v)(p)\leq \phi_{w,i,p}^{\lambda}(v)$\,;
\item $\phi_{w,i,p}^{\lambda}(w)=\becarre{F_i}(w)(p)$\,.
\end{enumerate}
\end{proposition}
\begin{proof}
Let $w\in \FS$, $i\in\ind$, $p\in\pp$ and $\pi\in\Pi$. 

1. The fact that $\phi_{w,i,p}^{\pi_\lambda(w)(i,p)}$ is affine follows readily from the definition (Eq.~\eqref{auxiliary}).

2. The monotonicity of $\phi_{w,i,p}^{\pi_\lambda(w)(i,p)}$ follows from the nonnegativity of $\pi_\lambda(w)(i,p)$.

3. Let $v\in\fr$. Since $w\in\FS$, there exists $(\lambda, \sigma,\mu, \gamma)\in \frp\times\Sigma[x]\times
\Sigma[x]^{n_i} \times \Sigma[x]^{n_0}$ such that $\deg (\sigma) \leq 2 m \deg T^i ,\ \deg (\langle\mu, r^i\rangle + \langle\gamma, r^0\rangle) \leq 2 m \deg T^i$: 
\[
\mybrackets{\becarre{F_i}(w)}(p)=p\circ T^i+\sum_{q \in \pp}\lambda(q) (w(q) - q) -\langle\mu, r^i\rangle -\langle\gamma, r^0\rangle +\sigma \,.
\]
Writing $\lambda=\pi_\lambda(w)(i,p)$, we get:
\[
\begin{array}{ll}
\displaystyle{\phi_{w,i,p}^{\lambda}(v)}&=\displaystyle{\sum_{q \in \pp}\lambda(q) v(q)-\sum_{q \in \pp}\lambda(q) w(q)+p\circ T^i+\sum_{q \in \pp}\lambda(q) (w(q) - q)}\\ 
&\displaystyle{-\langle\mu, r^i\rangle -\langle\gamma, r^0\rangle +\sigma}\\
&=\displaystyle{p\circ T^i+\sum_{q \in \pp}\lambda(q) (v(q) - q) -\langle\mu, r^i\rangle -\langle\gamma, r^0\rangle +\sigma} \,.
\end{array}
\]
Finally, 
\begin{equation}
\label{auxeqproof}
\phi_{w,i,p}^{\lambda}(v)-p\circ T^i-\sum_{q \in \pp}\lambda(q) (v(q) - q) +\langle\mu, r^i\rangle +\langle\gamma, r^0\rangle=\sigma \,,
\end{equation}
and recall that (Eq.~\eqref{optimalsolution}) 
\[
\begin{array}{ll}
 \mybrackets{\becarre{F_i}(w)}(p) =  &\displaystyle{\inf_{\lambda,\sigma, \mu,\gamma,\eta}} \eta\\
&\st
\left\{
\begin{array}{l}
\displaystyle{\eta- p\circ T^i-\sum_{\mathrm{q \in \pp}}\lambda(q) (w(q) - q) + \langle\mu, r^i\rangle + \langle\gamma, r^0\rangle} = \sigma\\
\lambda\in\frp,\ \sigma\in\Sigma[x] ,\ \mu\in\Sigma[x]^{n_i},\ \gamma\in\Sigma[x]^{n_0},\ \eta\in \rr \,,\\
\deg (\sigma) \leq 2 m \deg T^i ,\ \deg (\langle\mu, r^i\rangle + \langle\gamma, r^0\rangle) \leq 2 m \deg T^i \,.
\end{array}
\right.
\end{array}
\]
From Eq.~\eqref{auxeqproof}, $(\lambda,\sigma,\mu,\gamma,\phi_{w,i,p}^{\lambda}(v))$ is a feasible solution 
of the latter minimization problem and we conclude that $\mybrackets{\becarre{F_i}(v)}(p)\leq \phi_{w,i,p}^{\pi_\lambda(w)(i,p)}(v)$.

4.
\[
\phi_{w,i,p}^{\pi_\lambda(w)(i,p)}(w)=\sum_{q \in \pp}\lambda(q) w(q)+\mybrackets{\becarre{F_i}(w)}(p)-\sum_{q \in \pp}\lambda(q) w(q)
=\mybrackets{\becarre{F_i}(w)}(p)\enspace.
\]
\end{proof}
The properties presented in Prop.~\ref{phiauxproperty} imply some useful properties for the maps $\Phi_w^{\pi(w)}$.
\begin{proposition}[Properties of $\Phi_w^{\pi(w)}$]
\label{phiproperty}
Let $\pi\in\Pi$ and $w\in\FS$. The following properties are true:
\begin{enumerate}
\item $\Phi_w^{\pi(w)}$ is monotone on $\fr$\,;
\item $\becarre{F}\leq_\F \Phi_{w}^{\pi(w)}$\,;
\item  $\Phi_w^{\pi(w)}(w)=\becarre{F}(w)$ \,;
\item Suppose that the least fixpoint of $\Phi_w^{\pi(w)}$ is $L\in\frr$. Then $L$ can be computed as the unique optimal solution of the linear program:
    \begin{equation}
    \label{linearprogram}
    \inf\left\{\sum_{p'\in \pp} v(p')\mid \forall\, (i,p)\in \ind\times\pp,\ \phi_{i,w,p}^{\pi_\lambda(w)(i,p)}(v) \leq v(p),\ \forall q \in \pp, \becarre{{\xin}} (q) \leq v(q)\right\}\,.
    \end{equation}     
\end{enumerate}
\end{proposition}

LP problem~\eqref{linearprogram} corresponds exactly to the linear program presented in the case of quadratic templates~\cite[Eq. 4.4]{DBLP:journals/corr/abs-1111-5223}.
\begin{proof}
Let $\pi\in\Pi$ and $w\in\FS$.

1. The map $\Phi_w^{\pi(w)}$ is monotone as the map $\phi_{w,i,p}^{\pi_\lambda(w)(i,p)}$ is monotone for all $i\in\ind$ and 
for all $p\in\pp$, and the the fact that the point-wise supremum of monotone maps is also monotone. 

2.  Let $v\in\fr$ and let $p\in\pp$. Recall that:
\[
\mybrackets{\becarre{F}(v)}(p) = 
  \displaystyle{\sup\left\{\sup_{i\in\ind}
      \mybrackets{\becarre{F_i}(v)}(p),\becarre{{\xin}}(p)\right\}} \,,
\]
and from the third assertion of Prop.~\ref{phiauxproperty}, we have for all $i\in\ind$, 
$\becarre{F_i}(v)(p)\leq \phi_{w,i,p}^{\pi_\lambda(w)(i,p)}$, by taking the supremum over $\ind$ and then the supremum with
$\becarre{{\xin}}(p)$, we obtain that $\becarre{F}(v)(p)\leq \Phi_w^{\pi(w)}(v)(p)$, yielding the desired result.

3. This result follows readily from the fourth assertion of Prop.~\ref{phiauxproperty} and the definition of $\Phi_w^{\pi(w)}$ (Eq.~\eqref{auxiliaryfixpoint}). 

4. By Tarski's theorem and as $\Phi_w^{\pi(w)}$ is monotone, $\Phi_w^{\pi(w)}$ has a least fixpoint in $\fr$. Let 
$L$ be this least fixpoint supposed to be finite valued. Now, from Tarski's theorem and the definition of $\Phi_w^{\pi(w)}$, we have:
\[
\begin{array}{ll}
L&=\inf\{v\mid \Phi_w^{\pi(w)}(v)\leq_\F v\}\\
 &=\inf\left\{v\mid \forall\, (i,p)\in \ind\times\pp,\ \phi_{i,w,p}^{\pi_\lambda(w)(i,p)}(v) \leq v(p),\ \forall q \in \pp, \becarre{{\xin}} (q) \leq v(q)\right\}\enspace .
 \end{array}
\] 
Let us suppose that there exists a feasible solution $\bar{v}$ such that $\sum_{q\in\pp} \bar{v}(q)< \sum_{q\in\pp} L(q)$. Note that since $\becarre{{\xin}} \leq_\F \bar{v}$, $\sum_{q\in\pp} \bar{v}$ is finite. Then we have $\inf\{\bar{v},L\}\leq L$ and $\inf\{\bar{v},L\}\neq L$. As $\Phi_w^{\pi(w)}$ is monotone and as $\bar{v}$ and $L$ are feasible, we have $\Phi_w^{\pi(w)}(\inf\{\bar{v},L\})\leq \inf\{\bar{v},L\}$. This contradicts the minimality of $L$. We conclude that $L$ is the optimal solution of Linear Program~\eqref{linearprogram}. 
\end{proof}

\begin{remark}
We recall that the linear constraints in Problem~\eqref{linearprogram} come from the use of the function defined at Equation~\eqref{auxiliary} which is affine on the variable $v$. The linear forms are defined from the vector of Lagrange multipliers $\lambda$ 
found when we solve the minimization problem involved in Equation~\eqref{relaxedpoly}. If we had allowed a vector of SOS polynomials $\lambda$ as vector of Lagrange multipliers, we would obtain a set of polynomial inequalities that we would solve using SOS programming. The resulted problem would not have a feasible solution. 

For example, let us consider an SOS polynomial template $p$, an SOS (non scalar) polynomial $\lambda$ and a scalar $c$. 
Then, in this case, an analog of Problem~\eqref{relaxedpoly} would be: 
\[
\min\{v(p)\in \rr\mid \lambda(x) v(p)+c\leq v(p), \forall\, x\in \rr,\ v(p)\geq {\xin}^\dagger(p)\}
\]
We assumed that $p$ is a SOS polynomial template, implying that ${\xin}^\dagger(p)$ is strictly positive. Since $\lambda(x)$ is a non scalar SOS polynomial and $v(p)> 0$, then $v(p)(1-\lambda(x))-c$ is negative for some sufficiently large $x$. This proves the infeasibility of the problem.
\end{remark}

Recall that a function $g:\rd\mapsto\rr$ is upper-semicontinuous at $x$ iff for all $(x_n)_{n\in\nn}$ converging to $x$,
then $\limsup_{n\to+\infty} g(x_n)\leq g(x)$.
\begin{proposition}
Let $p\in\pp$. Then $w\mapsto \becarre{F}(w)(p)$ is upper-semicontinuous on $\FS\cap\frr$.
\end{proposition}
\begin{proof}
Let $\pi\in\Pi$, $w\in\FS\cap\frr$ and $p\in\pp$. Let $i\in\ind$. Let $(w_n)_{n\in\nn}$ be a sequence of elements of $\frr$ converging to $w$. Let $\lambda=\pi_\lambda(w)(i,p)$. Since $\phi_{i,w,p}^{\lambda}$ is affine on $\frr$, then $ \phi_{i,w,p}^{\lambda}$ is continuous on $\frr$ and finally $v\mapsto \Phi_w^{\pi(w)}(v)(p)$ is continuous on $\frr$ as a finite supremum of continuous functions on $\frr$. Then from the second point of Prop.~\ref{phiproperty}, for all $n\in\nn$, $\becarre{F}(w_n)(p)\leq \Phi_w^{\pi(w)}(w_n)(p)$.
By taking the $\limsup$, we obtain: $\limsup_{n\to +\infty} \becarre{F}(w_n)(p)\leq \limsup_{n\to +\infty}\Phi_w^{\pi(w)}(w_n)(p)=\Phi_w^{\pi(w)}(w)(p)=\becarre{F}(v)(p)$. 
\end{proof}
\subsection{Policy Iteration}

Next, we describe the policy iteration algorithm. We suppose that we have a post-fixpoint $w^0$ of $\becarre{F}$
in $\frr$. 

\begin{figure}[h] {\small
\begin{algorithm}[H]
  \SetKwInOut{Input}{input}\SetKwInOut{Output}{output}
 \Input{$w^0\in\frr$, a post-fixpoint of $\becarre{F}$}
 \Output{a fixpoint $w = \becarre{F}(w)$ if $\forall\, k\in\nn$, $w^k\in\FS$ or a post-fixpoint otherwise } 
 k=0\; \label{line1}
 \While{fixpoint not reached}{ \label{line2}
  \Begin(compute the next policy $\pi$ for the current iterate $w^k$){
    Compute $\becarre{F}(w^k)$ using Eq.~\eqref{relaxedfunctional} and Eq.~\eqref{relaxedpoly}\;\label{line4}
    \eIf{$w^k\in\FS$}{ 
    Define $\pi(w^k)$ \;\label{line5}
    }
    {
    return $w^k$\;
    }
  }
  \Begin(compute the next iterate $w^{k+1}$){
    Define $\Phi_{w^k}^{\pi(w^k)}$ and compute the least fixpoint $w^{k+1}$ of $\Phi_{w^k}^{\pi(w^k)}$ from Problem~\eqref{linearprogram}\;\label{line6} 
    k=k+1;\label{line7}
  }
 }

\end{algorithm}
}
 \caption{SOS-based policy iteration algorithm for PPS programs.}
 \label{pidefinition}
\end{figure}

We detail step by step the algorithm presented in
Figure~\ref{pidefinition}. At Line~\ref{line1}, the algorithm is initialized and thus $k=0$.
At Line~\ref{line4}, we compute $\rel{F}(w^k)$ using Eq.~\eqref{relaxedfunctional} and solve the SOS problem involved in Eq.~\eqref{relaxedpoly}.
At Line~\ref{line5}, if for all $i\in\ind$ and for $p\in\pp$, the SOS problem involved in Eq.~\eqref{relaxedpoly} has an optimal solution, 
then a policy $\pi$ is available and we can choose any optimal solution of SOS problem involved in Eq.~\eqref{relaxedpoly} as policy.
If an optimal solution does not exist then the algorithm stops and return $w^k$. Now, if a policy $\pi$ has been defined,
the algorithm goes to Line~\ref{line6} and we can define $\Phi_{w^k}^{\pi(w^k)}$ following Eq.~\eqref{auxiliaryfixpoint}.
Then, we solve LP problem~\eqref{linearprogram} and define the new bound on templates $w^{k+1}$ as the smallest fixpoint of
$\Phi_{w^k}^{\pi(w^k)}$. Finally, at Line~\ref{line7}, $k$ is incremented.

If for some $k\in\nn$, $w^{k}\notin \FS$ and $w^{k-1}\in\FS$ then the algorithm stops and returns $w^k$. Hence, we set for all $l\geq k$, $w^l=w^k$.
\begin{theorem}[Convergence result of the algorithm presented in Figure~\ref{pidefinition}]
\label{maintheorem}
The following statements hold:
\begin{enumerate}
\item For all $k\in\nn$, $w^k\in\frr$ and $\becarre{F}(w^k)\leq w^k$;
\item The sequence $(w^k)_{k\geq 0}$ generated by Algorithm~\ref{pidefinition} is decreasing and converges; 
\item Let $w^\infty=\lim_{k\to +\infty} w^k$, then $\becarre{F}(w^\infty)\leq w^\infty$. Furthermore, if 
for all $k\in\nn$, $w^k\in\FS$ and if $w^\infty\in\FS$ then $\becarre{F}(w^\infty)=w^\infty$. 
\end{enumerate}
\end{theorem}
\begin{proof}
1.  We reason by induction. We have $\rel{F}(w^0)\leq w^0$ and $w^0\in \frr$ by assumption. 
Now suppose that for some $k\in\nn$, $\rel{F}(w^k)\leq w^k$ and $w^k\in\frr$. If $w^k\notin\FS$ then $w^{l}=w^k$ for all $l\geq k$ and then we have proved the result. Now suppose that 
$w^k\in\FS$ and let us take $\pi\in\Pi$ such that $\Phi_{w^k}^{\pi(w^k)}(w^k)=\becarre{F}(w^k)$. From induction property $\Phi_{w^k}^{\pi(w^k)}(w^k)\leq w^k$ and thus $w^k$ is a post-fixpoint of $\Phi_{w^k}^{\pi(w^k)}$ belonging to $\frr$. Since every post-fixpoint of $\Phi_{w^k}^{\pi(w^k)}$ is greater than $\becarre{{\xin}}$ then least fixpoint of $\Phi_{w^k}^{\pi(w^k)}$ is finite valued and thus it is 
the optimal solution $w^{k+1}$ of Problem~\eqref{linearprogram}. Moreover from the second point of Prop.~\ref{phiproperty}, $\becarre{F}(w^{k+1})\leq_\F \Phi_{w^k}^{\pi(w^k)}(w^{k+1})$ and since $w^{k+1}$ is the least fixpoint of $\Phi_{w^k}^{\pi(w^k)}(w^{k+1})$ then $\becarre{F}(w^{k+1})\leq w^{k+1}$. This completes the proof and for all $k\in\nn$, $w^k\in\frr$ and $\becarre{F}(w^k)\leq w^k$.

2. Let $k\in\nn$. If $w^k\notin \FS$ then $w^{k+1}=w^k\leq w^k$. 
Now suppose that $w^k\in \FS$ and let $\pi\in\Pi$ such that $\Phi_{w^k}^{\pi(w^k)}(w^k)=\becarre{F}(w^k)$, then from the third point of Prop.~\ref{phiproperty}, $\Phi_{w^k}^{\pi(w^k)}(w^k)=\becarre{F}(w^k)\leq w^k$; the inequality results from the first assertion. Then $w^k$ is a post-fixpoint of $\Phi_{w^k}^{\pi(w^k)}$. Since $w^{k+1}$ is the least fixpoint of $\Phi_{w^k}^{\pi(w^k)}$ and $\Phi_{w^k}^{\pi(w^k)}$ is monotone then from 
Tarski's theorem $w^{k+1}\leq w^k$. From the first point, for all $k\in\nn$, $w^k\in\frr$. Moreover by definition of $\becarre{F}$, $\becarre{{\xin}}\leq \becarre{F}(w^k)$ for all $k\in \nn$, then 
from the first point $(w^k)_{k\geq 0}$ is lower bounded then it converges to some $w^\infty$. 

3. If for some $k$, $w^{k}\notin \FS$ and $w^{k-1}\in\FS$, then $w^\infty=w^k$ and we have $\becarre{F}(w^\infty)\leq w^\infty$
from the first point. Now suppose that for all $k\in\nn$, $w^k\in\FS$. Since $\becarre{F}$ is monotone then for all $k\in\nn$, $\becarre{F}(w^\infty)\leq \becarre{F}(w^k)\leq w^k$ from the first point. Now taking the limit of the right-hand side, we get ${F}(w^\infty)\leq w^{\infty}$. Now, let $k\in\nn$ and let $\pi\in\Pi$ such that $\Phi_{w^k}^{\pi(w^k)}(w^k)=\becarre{F}(w^k)$. From the second point, $w^{k+1}\leq w^k$ and from the monotonicity of $\Phi_{w^k}^{\pi(w^k)}$, we have $w^{k+1}=\Phi_{w^k}^{\pi(w^k)}(w^{k+1})\leq \Phi_{w^k}^{\pi(w^k)}(w^k)=\becarre{F}(w^k)$. By taking the $\limsup$ on $k$, we get $w^\infty\leq \limsup_{k\to +\infty}\becarre{F}(w^k)$. As $\becarre{F}$ is upper-semicontinuous on $\FS\cap \frr$ then, if $w^\infty\in\FS$, $w^\infty\leq \limsup_{k\to +\infty}\becarre{F}(w^k)\leq \becarre{F}(w^\infty)$ and so $w^\infty=\becarre{F}(w^\infty)$.
\end{proof}
\subsection{Initialization and templates choice}
\label{initsub}
In Section~\ref{sec:basicinductive-domainsSOS}, we have made the assumption that the template basis was given by an oracle. 
Moreover, in Algorithm~\ref{pidefinition}, we suppose that we have a post-fixpoint $w^0\in\frr$ of $\becarre{F}$. 
Now, we give details about the templates basis choice and the computation of a post-fixpoint 
$w^0 \geq w^\infty$. The templates basis choice relies on the computation of a template basis composed of one element.   
This single template is constructed by the method developed in~\cite{SAS1} and
is then completed using the strategy proposed in~\cite[Ex. 9]{SAS1}. The single template computation also permits us 
to compute $w^0$.
Actually, the method developed in~\cite{SAS1} is constructed by using the definition of being a post-fixpoint of $\becarre{F}$. 
Indeed, suppose that the templates basis is constituted of one template $p$ then $w^0$ is a post-fixpoint 
$\becarre{F}$ if and only if $\becarre{F}(w^0)(p)\leq w^0(p)$. This is equivalent to: 
\[
\becarre{{\xin}}=\inf\{\eta\mid \eta-p+\sum_{j=1}^{n_{\mathrm{in}}} \nu_j^{\mathrm{in}}
r_j^{\mathrm{in}}\in\Sigma[x],\
\nu^{\mathrm{in}}\in\Sigma[x]^{n_\mathrm{in}}\}\leq w^0 \,,
\]
and for all $i\in\ind$:
\[
\begin{array}{lll}
 \mybrackets{\becarre{F_i}(w^0)}(p) =  &\displaystyle{\inf_{\lambda,\mu,\gamma,\eta}} \eta &\leq  w^0 \,.\\
&\st
\left\{
\begin{array}{l}
\displaystyle{\eta- p\circ T^i-\lambda^i(w^0 - p) + \langle\mu, r^i\rangle + \langle\gamma, r^0\rangle} \in\Sigma[x]\\
\lambda^i\geq 0,\ \mu\in\Sigma[x]^{n_i},\ \gamma\in\Sigma[x]^{n_0},\ \eta\in \rr
\end{array}
\right.&
\end{array}
\]  
By definition of the infimum, it is equivalent to the existence of $\nu^{\mathrm{in}}\in\Sigma[x]$ and 
for all $i\in\ind$ of $\lambda^i\geq 0$, $\mu^i\in\Sigma[x]^{n_i}$, 
$\gamma^i\in\Sigma[x]^{n_0}$ such that:
\begin{equation}
\label{templateeq}
\begin{array}{c}
w^0-p+\sum_{j=1}^{n_{\mathrm{in}}} \nu_j^{\mathrm{in}} r_j^{\mathrm{in}}\in\Sigma[x]\\
w^0- p\circ T^i-\lambda^i(w^0 - p) + \langle\mu, r^i\rangle + \langle\gamma, r^0\rangle \in\Sigma[x]
\,.
\end{array}
\end{equation}
Now to find a template, it suffices to find $p$ such that Eq.~\eqref{templateeq} holds. 
However, the following two issues remain. 

First, without an objective function, $p=0$ is a solution of Eq.~\eqref{templateeq}.
A workaround to avoid this trivial solution consists of optimizing a certain objective function under the constraints given in Eq.~\eqref{templateeq}. In~\cite{SAS1}, a similar optimization procedure (Problem (13) of~\cite{SAS1}) is used to prove a property of the form $\rea\subseteq \{x\in\rd\mid \kappa(x)\leq \alpha\}$, for a given real-valued function $\kappa$.
Here, we are interested in proving the boundedness of the reachable value set, which corresponds to minimize $\alpha$ with $\kappa=\norm{\cdot}_2^2$.

Second, finding $\lambda^i$ and $p$ satisfying Eq.~\eqref{templateeq} boils down to solving a bilinear SOS problem, which is not easy to handle in practice. Thus, we fix $\lambda^i=1$ as in Lyapunov equations. 
We also take $w^0=0$ since $p$ has a constant part.
Finally, to obtain a template $p$, we solve the following SOS problem:
\begin{align}
\label{polsynthesis}
\begin{aligned}
\inf_{p\in \rr[x]_{2m}, w\in\rr} & \quad w \enspace, \\[-1em]			 
\text{s.t.}  & \quad - p = \sigma_0 - \sum_{j=1}^{n_{\mathrm{in}}} \sigma_jr_j^{\mathrm{in}}  \enspace , \\[-1em]
& \quad \forall\, i\in\ind,\ \displaystyle{p-p \circ T^i= \sigma^i - \sum_{j=1}^{n_i}\mu_j^i r_j^i - \sum_{j=1}^{n_0}\gamma_j^i r_j^0}  \enspace , \\[-1em]
& \quad \displaystyle{w + p -\norm{\cdot}_2^2 = \psi} \enspace , \\
& \quad \forall\, j=1,\ldots, n_{\mathrm{in}} \enspace,\ \sigma_j\in\Sigma[x]\enspace ,\ \deg (\sigma_j r_j^{\mathrm{in}})  \leq 2m\enspace,\\
& \quad \sigma_0\in\Sigma[x]\enspace ,\ \deg (\sigma_0)  \leq 2m\enspace,\\
& \quad \forall\, i\in\ind \enspace ,\ \sigma^i\in \Sigma[x]\enspace ,\ \deg (\sigma^i) \leq 2 m \deg T^i \enspace,\\
& \quad \forall\, i\in\ind \enspace ,\ \forall\, j=1,\ldots, n_i \enspace,\ \mu_j^i\in\Sigma[x]\enspace ,\ \deg (\mu_j^i r_j^i)  \leq 2 m \deg T^i \enspace ,\\
& \quad \forall\, i\in\ind\enspace ,\ \forall\, j=1,\ldots, n_0\enspace ,\ \gamma^i\in \Sigma[x]\enspace ,\ \deg (\gamma_j^i r_j^0)  \leq 2 m \deg T^i  \enspace ,\\
& \quad \psi \in \Sigma[x] \enspace,\ \deg (\psi) \leq 2 m  \enspace .\\
\end{aligned} 
\end{align}
Let $(p,w)$ be a solution of Problem~\eqref{polsynthesis}. In ~\cite[Prop. 1]{SAS1}, we proved that the set $\{x\in\rd\mid p(x)\leq 0\}$ defines an inductive invariant. To complete the template basis, we use the strategy proposed in~\cite[Ex. 9]{SAS1}, that is, we work with the templates basis $\{x\mapsto x_i^2, i\in\iset{d}\}\cup\{p\}$. We thus use the inductive invariant set $\{p(x)\leq 0, x_i^2\leq w\}$
as initialization i.e. the initial bound is $w^0(q)=w$ if $q\neq p$ and $w^{(0)}(q)=0$ if $q=p$. 
As opposed to the approach of~\cite{SAS1}, we avoid increasing the degree of polynomial $p$ to obtain better bounds on the reachable values set.
\paragraph{Computational considerations}
 The number of (a-priori unknown) coefficients of the polynomial $p$ (of degree $2m$ and $d$ variables) appearing in Problem~\ref{polsynthesis} is $\binom{2 m + d}{d}$. Similarly, the number of coefficients of each $\sigma^i$ (resp.~$\mu_j^i$ and $\gamma_j^i$) is $\binom{2 m \deg T^i + d}{d}$.
Thus, Problem~\ref{polsynthesis} can be reformulated as an SDP program involving $\binom{2 m + d}{d} + \sum_{i \in \ind} [1+n_i + n_0]  \binom{2 m \deg T^i + d}{d}$ SDP variables. 
Therefore, our framework is expected to be tractable when either $d$ or $m$ is small.
As mentioned in~\cite[Section 4]{SAS1}, one could address bigger instances while exploiting sparsity properties of the initial system, as in~\cite{Waki06sumsof}.
%

\section{Experiments}
\label{sec:exampleSOS}
\subsection{Details of the running Example.}
Recall that our running example is given by the following PPS:
$(\xin,X^0$, $\{X^1,X^2\}$, $\{T^1,T^2\})$, where:\\
\[
\begin{array}{lcl}
  \begin{array}{l}
    \xin= [-1, 1] \times [-1, 1]\\
    X^0=\rr^2
  \end{array}
  &  \text{ and }& 
  \left\{
    \begin{array}{l}
      X^1=\{x\in\rr^2\mid -x_1^2+1\leq 0\} \\
      X^2= \{x\in\rr^2\mid x_1^2-1< 0\}    
    \end{array}
  \right.
\end{array}
\]
  \\
\noindent
and the functions relative to the partition $\{X^1,X^2\}$ are:
\[
\begin{array}{c}
T^1(x_1,x_2) =\begin{pmatrix} 
0.687x_1+0.558x_2-0.0001x_1 x_2 \\ 
-0.292x_1+0.773x_2
\end{pmatrix}\\
\text{and}\\
T^2(x_1,x_2) =\begin{pmatrix} 
0.369x_1+0.532x_2-0.0001x_1^2\\
-1.27x_1+0.12x_2-0.0001x_1x_2
\end{pmatrix}
\,.
\end{array}
\]
The first step consists in constructing the template basis and compute the template $p$ and bound $w$ on the reachable values as a solution of Problem~\eqref{polsynthesis}. We fix the degree of $p$ to 6. 
The template $p$ generated from Matlab is of degree 6 and is equal to
 \[
 \begin{array}{l}
 -1.931348006+3.5771x_1^2+2.0669x_2^2+0.7702x_1x_2-(2.6284\text{e--}4 )x_1^3\\
 -(5.5572\text{e--}4)x_1^2x_2+(3.1872\text{e--}4)x_1x_2^2+0.0010x_2^3-2.4650x_1^4-0.5073x_1^3x_2\\
 -2.8032x_1^2x_2^2-0.5894x_1x_2^3-1.4968x_2^4+(2.7178\text{e--}4)x_1^5+(1.2726\text{e--}4)x_1^4x_2\\
 -(3.8372\text{e--}4)x_1^3x_2^2+(6.5349\text{e--}5)x_1^2x_2^3+(5.7948\text{e--}6)x_1x_2^4-(6.2558\text{e--}4)x_2^5\\
 +0.5987x_1^6-0.0168x_1^5x_2+1.1066x_1^4x_2^2+0.3172x_1^3x_2^3+0.8380x_1^2x_2^4+0.0635x_1x_2^5\\
 +0.4719x_2^6 \,.
 \end{array}
 \]
The upper bound $w$ is equal to $2.1343$. As suggested in Section~\ref{initsub}, we can take the template basis $\pprun=\{p,x\mapsto x_1^2, x\mapsto x_2^2\}$.  We write $q_1$ for $x\mapsto x_1^2$ and $q_2$ for $x\mapsto x_2^2$.
The basic semi-algebraic $\{x\in\rr^2\mid p(x)\leq 0,\ q_1(x)\leq 2.1343,\ q_2(x)\leq 2.1343\}$ is an inductive invariant and the corresponding bounds function is $w^0=(w^0(q_1),w^0(q_2),w^0(p))=(2.1343,2.1343,0)$.

As in Line~\ref{line4} of Algorithm~\ref{pidefinition}, we compute the image of $w^0$ by $\becarre{F}$ using SOS (Eq.~\eqref{relaxedpoly}). We found that 
\[
\becarre{F}(w^0)(q_1)=1.5503,\ \becarre{F}(w^0)(q_2)=1.9501 \text{ and } \becarre{F}(w^0)(p)=0\enspace .
\]
Since $w^0\in\FSrun$, Algorithm~\ref{pidefinition} goes to Line~\ref{line5} and the computation of $\becarre{F}(w^0)$ permits to determine a new policy $\pi(w^0)$. The important data is the vector $\lambda$. For example, 
for $i=1$ and the template $q_1$, the vector $\lambda$ is $(0,0,2.0331)$. It means that we associate
for each template $q$ a weight $\lambda(q)$. In the case of $\lambda=(0,0,2.0331)$, $\lambda(q_1)=0$, 
$\lambda(q_2)=0$ and $\lambda(p)=2.0332$. For $i=1$, the template $q_1$ and the bound vector $w^0$, 
the function $\phi_{w^0,1,q_1}^{\lambda}(v)=2.0331 v(p)+1.5503$. 

To get the new invariant, Algorithm~\ref{pidefinition} goes to Line~\ref{line6} and we compute a bound vector $w^1$ solution of Linear Program~\eqref{linearprogram}.
In this case, it corresponds to the following LP problem:
 \[
 \underset{
 1\leq v(q_1),\ 1\leq v(q_2),\ 0\leq v(p),\qquad (\text{init})}{
  \underset{
 0.4578 v(p)+0.8843\leq v(q_1),\ 0.2048v(p)+1.9501\leq v(q_2),\ 0.9985 v(p)-3.4691\text{e--}7\leq v(p),\qquad (i=1)}{
 \underset{
 2.0331 v(p)+1.5503\leq v(q_1),\ 1.0429 v(p)+1.2235\leq v(q_2),\ 0.9535 v(p)-0.0248\leq v(p),\qquad (i=2)}
 {\operatorname{\min} v(q_1)+v(q_2)+v(p)}}}
 \]
We obtain:
\[
w^1(q_1)=1.5503,\ w^1(q_2)=1.9501 \text{ and } w^1(p)=0\enspace.
\]
Then, we come back to Line~\ref{line4} of Algorithm~\ref{pidefinition} and we compute $\becarre{F}(w^1)$ using the SOS program Eq.~\eqref{relaxedpoly}. The implemented stopping rule is $\norm{\becarre{F}(w^k)-w^k}_{\infty}\leq 1\text{e--}6$ and 
since $\norm{\becarre{F}(w^1)-w^1}_{\infty}\leq 1\text{e--}6$,
Algorithm~\ref{pidefinition} terminates. The computed sets are presented in
Figure~\ref{invs}, page~\pageref{invs}. Figure~\ref{fig:more} presents the
semialgebraic sets obtained with higher dimensional templates, up to degree
10. Results are similar but could lead to different numbers of iterations
depending on the degree. In case of multiple iterations, the final value is also
reached at iterate 1 and is slightly modified by following iterations.

\begin{figure}[h]
  \begin{center}
    \includegraphics[width=.5\textwidth]{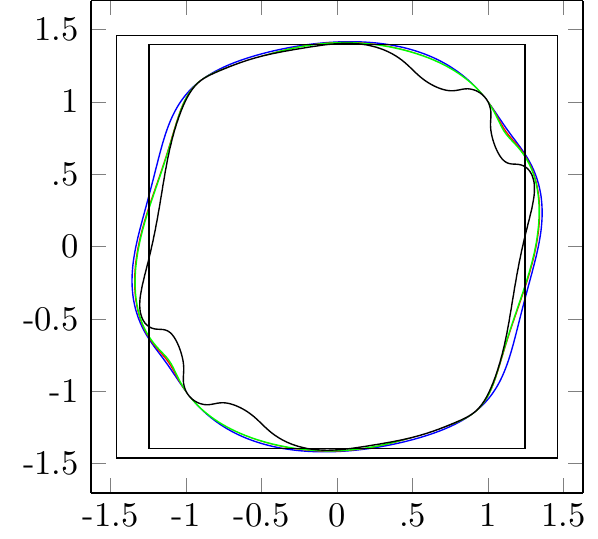}
  \end{center}
The semialgebraic sets denotes templates of degree 7 (blue), 8 (red), 9 (green) and 10 (black). The first box denotes the initial bounds as obtained in figure~\ref{invs}. The
second one is the one obtained after 1 iterations. Except degree 7 that
converged in 4 iterations, all others converged in 1. Degree 10 faced numerical
issues and did not allow to refine the bounds without errors.
\caption{Different templates and associated bounds computed with Policy Iterations}
\label{fig:more}
\end{figure}
\subsection{Benchmarks.}
The presented analysis has been applied to available examples of the control
community literature: piecewise linear systems, polynomial systems, etc. We
gathered the examples matching our criteria: discrete systems, possibly
piecewise, at most polynomial. In all the considered cases, no common quadratic
Lyapunov existed. In other words, not only the existing linear abstractions such
as intervals or polyhedra would fail in computing a non trivial post-fixpoint,
but also the existing analyses dedicated to digital filters such
as~\cite{DBLP:conf/esop/Feret04,DBLP:conf/esop/GawlitzaS07,DBLP:journals/corr/abs-1111-5223,DBLP:conf/hybrid/RouxJGF12}.

The analysis has been implemented in Matlab and relies on the Mosek SDP
solver~\cite{mosek}, through the Yalmip~\cite{YALMIP} SOS front-end. Without
outstanding performances, all experiments are performed within a few seconds per
iteration, which makes us believe that a more serious implementation would
perform better. We recall that the analysis could be interrupted at any point,
still providing a safe upper bound.

We next present the examples handled by our SOS policy iteration algorithm: 

\begin{example}
\label{example1}
The following example corresponds to~\cite[Ex. 2.1]{DBLP:journals/tac/Feng02} and represents a piecewise linear system with 2 cases handling 3 variables. 
The initial set is:
\[
\xin = [-1, 1]^3\enspace.
\]
The set where the state-variable lies is:
\[ 
X^0 = \rr^3\enspace .
\]
The sets defining the partition of the state-space are:
\[
\begin{array}{lr}
X^1 = \{ (x,y,z) \in \rr^3 | x \leq 0 \}, & 
X^2 = \{ (x,y,z) \in \rr^3 | x > 0 \}\enspace .
\end{array}
\]
Finally the dynamics associated to the partition are:
\[
\begin{array}{lr}
T^1(x,y,z) = \begin{pmatrix} x + 0.5 y\\ -0.3 x+0.8 y\\0.4 z\end{pmatrix}, &
T^2(x,y,z) = \begin{pmatrix}x+.4 y+0.01 z\\ -0.1 x+0.8 y\\0.5 z\end{pmatrix}\enspace .
\end{array}
\]
\end{example}

\begin{example}
\label{example2}
We consider the example~\cite[Ex. 3.3]{DBLP:journals/tac/Feng02} which describes a piecewise linear system with 4 cases handling 2 variables.
The initial set is:
\[
\xin = [-1, 1]^2\enspace .
\]
The set where the state-variable lies is:
\[ 
X^0 = \rr^2\enspace .
\]
The sets defining the partition of the state-space are:
\[
\begin{array}{lr}
X^1 = \{ (x,y) \in \rr^2 | x \leq -1 \},& 
X^2 = \{ (x,y) \in \rr^2 | x \in ]-1,1] \land y > 0 \},\\
X^3 = \{ (x,y) \in \rr^2 | x \in ]-1,1] \land y \leq 0 \},&
X^4 = \{ (x,y) \in \rr^2 | x > 1 \}\enspace .
\end{array}
\]
Finally the dynamics associated to the partition are:
\[
\begin{array}{lr}
T^1(x,y) = \begin{pmatrix}0.9 x-0.01 y\\ 0.1 x+y -0.02\end{pmatrix},&
T^4(x,y) = \begin{pmatrix}0.9 x-0.01 y\\ 0.1 x+y + 0.02\end{pmatrix},\\
\multicolumn{2}{c}{T^2(x,y) = T^3(x,y) = \begin{pmatrix}x-0.02 y\\ 0.02 x+0.9 y\end{pmatrix}}\enspace .\\
\end{array}
\]
\end{example}


\begin{example}
\label{example3}
The following example is the piecewise quadratic system with 2 cases handling 2 variables~\cite[Ex. 3]{DBLP:conf/cdc/AhmadiJ13}.
The initial set is:
\[
\xin = [-1, 1]^2\enspace .
\]
The set where the state-variable lies is:
\[ 
X^0 = \rr^2\enspace .
\]
The sets defining the partition of the state-space are:
\[
\begin{array}{lr}
X^1 = \{ (x,y) \in \rr^2 | -x^4+x^2-1 \leq 0 \},&
X^2 = \{ (x,y) \in \rr^2 | x^4-x^2 +1< 0 \}\enspace .
\end{array}
\]
Finally the dynamics associated to the partition are:
\[
\begin{array}{lr}
T^1(x,y) = \begin{pmatrix}0.687 x+0.558 y-0.0001 x y\\-0.292 x+0.773 y\end{pmatrix},&
T^2(x,y) = \begin{pmatrix}0.369 x+0.532 y-0.0001 x^2\\-1.27 x+0.12 y-0.0001 x y\end{pmatrix}\enspace .
\end{array}
\]
\end{example}

\begin{example}
\label{example4}
The following example is the hand-crafted piecewise polynomial of degree 3 with 2 cases developed in~\cite{SAS1}.
The initial set is:
\[
\xin = [0.9, 1.1] \times [0, 0.2]\enspace .
\]
The set where the state-variable lies is:
\[ 
X^0 = \rr^2\enspace .
\]
The sets defining the partition of the state-space are:
\[
\begin{array}{lr}
X^1 = \{ (x,y) \in \rr^2 | x^2 + y^2 \leq 1  \},&
X^2 = \{ (x,y) \in \rr^2 | x^2 + y^2 > 1 \}\enspace .
\end{array}
\]
Finally the dynamics associated to the partition are:
\[
\begin{array}{lr}
T^1(x,y) = \begin{pmatrix}x^2 + y^3\\ x^3 + y^2\end{pmatrix},&
T^2(x,y) = \begin{pmatrix}0.5   x^3 + 0.4   y^2\\ -0.6   x^2 + 0.3   y^2\end{pmatrix}\enspace .
\end{array}
\]
\end{example}
The table~\ref{tab:exp} summarizes the examples considered, the bounds obtained,
the degree of the polynomial templates and the number of iterations
performed before reaching the fixpoint.

\begin{table}[ht]
{  \centering
  \begin{tabularx}{\textwidth}{Xccc}
    \hline Examples & Bounds (ie. $x_i^2$) & Degree & \# it.\\ \hline
    \multirow{5}{6cm}{Running example~\ref{running}} & No good invariant & 4 & $-$\\
    &   $[1.5503, 1.9501]$ & 6 & 1 \\
    & $[1.5503,     1.9502]$ & 8 & 7 \\
    &  $[ 1.5500,  1.9436]$ & 10 & 1 \\
    & $[ 1.5503 ,   1.9383]$ & 12 & 2 \\ \hline
    \multirow{4}{6cm}{Example~\ref{example1}} & $[
    3.8260, 2.1632, 1.0000 ]$ & 4 &
    1\\
    & $ [3.7482, 1.8503, 1.0000 ]$ & 6 &
    1\\
    & No good invariant & 8,10,12 & $-$
    \\ \hline
    \multirow{4}{6cm}{Example~\ref{example2}} & $[
    1.8359 , 1.3341]$ & 4 & 2
    \\
    & $[1.5854, 1.2574]$ & 6 & 5
    \\
    & $[1.5106 , 1.2569]$ & 8 & 4
    \\
    & $[1.4813, 1.2544]$ & 10 & 6   \\ \hline
    \multirow{5}{6cm}{Example~\ref{example3}} & $[1.5624,1.2396]$ & 4 & $3$\\
    &   $[1.5581, 1.1764]$ & 6 & 1 \\
    & $[1.5531,     1.1511]$ & 8 & 1 \\
    &  No good invariant & 10,12 & $-$ \\ \hline
    \multirow{2}{6cm}{Example~\ref{example4}} &
    No good invariant & 4,6,8,10 & $-$\\
    &  $[ 1.2100, 0.9989]$ & 12 & max (10)  \\

    \hline

  \end{tabularx}
}
  ``No good invariant'' occurs when the template synthesis fails, i.e. does not
  provide a sound post-fixpoint or some numerical issues occurs during the policy iterations phase. 
  It seems to be due to the large size of the SOS problems together with numerical issues related to the interior
  point methods implemented in the relying solvers.
  \caption{Experiments}
  \label{tab:exp}
\end{table}


\section{Conclusion}
\label{sec:concluSOS}
We proposed an extension of policy iteration algorithms, using Sum-of-Squares programming. 
This extension allows to consider the wider class of disjunctive polynomial programs. 
In this new setting, we showed that we keep the advantage of policy iteration algorithms,
while producing a sequence of increasingly safe over-approximations of the reachability set.

As future work, we plan to generalize this algorithm to programs involving non-polynomial updates, including square roots, divisions as well as transcendental functions. The computational method developed in the paper could be also generalized to other classes of nonlinear switched systems involving either random or temporal switching.

\section*{Acknowledgments}

The research leading to these results has partly received funding from the
European Research Council under the European Union’s Seventh Framework Programme
(FP/2007-2013) / ERC Grant Agreement nr. 306595 ``STATOR'', the LabEx PERSYVAL-Lab (ANR-11-LABX-0025-01) funded by the French program ``Investissement d'avenir'',  the RTRA/STAE
BRIEFCASE project grant, the ANR projects INS-2012-007 CAFEIN, and ASTRID
VORACE.

\bibliographystyle{alpha}
\bibliography{policySOSbib}


\end{document}